\documentclass{article}
\usepackage{ijcai24}

\usepackage{color}
\usepackage{times}
\usepackage{soul}
\usepackage{url}
\usepackage[hidelinks]{hyperref}
\usepackage[utf8]{inputenc}
\usepackage[small]{caption}
\usepackage{graphicx}
\usepackage{amsmath}
\usepackage{amsthm}
\usepackage{booktabs}
\usepackage{algorithm}
\usepackage{svg}

\usepackage{algpseudocode}
\algrenewcommand\algorithmicif{\textbf{if}}
\newlength{\myindent}
\algtext*{EndFor}% Disable "end for" text
\algtext*{EndIf}% Disable "end for" text
\newcommand{\InlineIf}[2]{%
    \State \algorithmicif\ #1\ \algorithmicthen\ #2%
}
\usepackage[switch]{lineno}
\usepackage{amsthm}
\usepackage{amsmath}
\usepackage{mathtools}
\usepackage{amsfonts}
\newcommand{\defeq}{\vcentcolon=}
\usepackage{float} 

\newtheorem{example}{Example}
\newtheorem{theorem}{Theorem}
\newtheorem{lemma}{Lemma}
\newtheorem{corollary}{Corollary}
\newtheorem{proposition}{Proposition}

\theoremstyle{definition}
\newtheorem{definition}{Definition}

\theoremstyle{remark}
\newtheorem{claim}{Claim}

\DeclareMathOperator*{\argmin}{arg\,min}

\title{Computing Optimal Equilibria in Repeated Games with Restarts}

\author{
    Ratip Emin Berker \And Vincent Conitzer\\
    \affiliations
    Foundations of Cooperative AI Lab (FOCAL), Computer Science Department, Carnegie Mellon University\\
    \emails
    \{rberker, conitzer\}@cs.cmu.edu 
}

\begin{document}

\maketitle

\begin{abstract}
 Infinitely repeated games can support cooperative outcomes that are not equilibria in the one-shot game.  The idea is to make sure that any gains from deviating will be offset by retaliation in future rounds. However, this model of cooperation fails in anonymous settings with many strategic agents that interact in pairs.  Here, a player can defect and then avoid penalization by immediately switching partners.  In this paper, we focus on a specific set of equilibria that avoids this pitfall.  In them, agents follow a designated sequence of actions, and \textit{restart} if their opponent ever deviates. We show that the socially-optimal sequence of actions consists of an infinitely repeating goal value, preceded by a \textit{hazing period}. We introduce an equivalence relation on sequences and prove that the computational problem of finding a representative from the optimal equivalence class is (weakly) NP-hard. Nevertheless, we present a pseudo-polynomial time dynamic program for this problem, as well as an integer linear program, and show they are efficient in practice. Lastly, we introduce a fully polynomial-time approximation scheme that outputs a hazing sequence with arbitrarily small approximation ratio. 
\end{abstract}

\section{Introduction}\label{sec:intro}
In social dilemmas, individual incentives hinder collective benefit: mutual cooperation is the best outcome for both players, but it is not a Nash equilibrium. Consider the symmetric two-player game in Table \ref{table:game}:
    \begin{table}[h]
            \centering
            \begin{tabular}{|c|c|c|c|}
            \hline & $D$& $C_1$ &$C_2$\\
            \hline
            $D$ & 4,4 & 11,0 & 14,0\\
            \hline $C_1$ & 0,11 & 5,5 & 0,0\\
            \hline $C_2$ & 0,14 & 0,0 & 8,8\\ \hline 
            \end{tabular}
            \caption{\label{table:game}Payoffs of a  two-player symmetric game (row, column)}
    \end{table}
    
    Notice that for each player, $D$ is the strictly dominant action, ensuring the maximum payoff against any fixed action of their opponent. This results in $(D,D)$ being the only Nash equilibrium. In the infinitely repeated version of this game, however, mutual cooperation can be achieved: Consider the grim-trigger strategy where both players agree to play $C_2$ in every round, but if their opponent defects, they switch to playing $D$ in all future rounds. During the game, a player can increase their payoff in a round by at most $14-8=6$ by defecting, but their payoff in all future rounds is now bound above by $4$, as their opponent will switch to $D$, resulting in a per-round net loss of at least $8-4=4$ compared to if they had stuck to $C_2$. If the players value future rounds sufficiently (i.e., have a `discount factor' close to 1, assumed true for the rest of this section), the gain of defecting ($6$, once) will be offset by loss due to the opponent retaliating ($4$, all future rounds), resulting in neither of the players defecting from $C_2$.
    
    This model of cooperation, however, fails in anonymous settings with many players, in which players can simply find a new partner to play with. 
    This setting can model any situation that involves two-player interactions within a larger pool, such as monogamous relationships, employer-employee interactions, and two-person research collaborations. Critically, each of these `partnerships' can last for arbitrarily many rounds, but can be terminated at any point by one of the partners, who can then find themselves a new partner in the larger pool. In this game-theoretic setting, a player can defect and avoid retaliation, if they are able to switch partners immediately following the defection. If there is no way for a player to check their partner's history (that is, players are \textit{anonymous}), this setting may result in the emergence of `serial defectors,' who perpetually defect on a partner and move on the next, even if all their partners would follow the `grim trigger' strategy if the relationship continued. This is especially relevant for settings where the players are AI agents (such as trading bots), who might more easily conceal their identity compared to traditional human players.

    To avoid this pitfall, we turn to a specific type of equilibria in such infinitely repeated two-player games, those with \textit{restarts}. Instead of the grim-trigger strategy of `defect forever once the opponent deviates', consider instead a strategy profile where all players agree on a planned sequence of actions to follow, and they restart the same sequence with a new partner if their opponent ever deviates from it. (In the context of our paper, the punishment strategy of simply leaving the relationship, thereby forcing the partner to restart as well, is in fact without loss of generality: no other punishment strategy could be more effective, because if it were more effective, the partner would simply leave the relationship.)\footnote{The type of equilibrium we study remains an equilibrium in the standard repeated game setting (without restarts), and so is a {\em refinement} of the traditional concept.  In this case, if one's partner deviates, one continues with the same partner but plays {\em as if} one were starting from the beginning, which is sufficient to deter deviation.}

    For instance, going back to the game in Table \ref{table:game}, say the two players agree on the sequence $(D,C_1,C_2,C_2,C_2,\ldots)$. If the players stick to the plan, the sequence of payoffs they receive is $(4,5,8,8,8,\ldots)$. Neither will deviate from the plan in the first round ($D$), as doing so can only lower their payoff. In the second round $(C_1)$, deviating brings a payoff of $11$, but results in the opponent ending the relationship and having to go back to the start of the sequence with someone else.  This results in a per-round average payoff of $(4+11)/2=7.5$ for the deviating player, which is less than the per-round payoff of 8 she can eventually receive by sticking to the planned sequence. Similarly, once the players get to the $(C_2)$ portion of their sequence, any deviation can bring at most a payoff of $14$ (an increase of $6$), but results in a restart where the next two rounds (4,5) does a total damage of $(8-4)+(8-5)=7$ when compared to sticking with the sequence, making the overall sequence stable. Conceptually, the planned sequence consists of a `hazing' period $(D,C_1)$, followed by lasting socially-optimal cooperation $(C_2,C_2,...)$.
    
    Not every planned sequence is stable: an alternative plan $(C_1,C_2,C_2,C_2,\ldots)$ would incentivize repeatedly deviating on the first step, to obtain 11 per round. Another alternative, $(D,C_2,C_2,C_2,\ldots)$, would see a player repeatedly deviating on the second step, ensuring a per round utility of $(4+14)/2=9>8$. Both of these alternatives suffer from under-hazing, as the cost of a restart for a defecting player fails to offset the gains from defection.
    
$(D,C_1,C_2,C_2,\ldots)$ is not the only stable sequence; so is $(D,C_1, C_1,\ldots,C_1,C_2,C_2,\ldots)$, for an arbitrarily large (positive) number of $C_1$s. However, this results in unnecessarily delaying the socially optimal outcome of $C_2$, i.e., over-hazing. Hence, in this paper we ask: how can we optimize the payoff of a planned sequence, while ensuring its stability? 

\subsection{Related Work}

Repeated games (without restarts) have long been of interest to the AI community.  In contrast to one-shot games, they introduce a temporal component to the game, and they allow modeling settings where cooperation can be sustained thanks to the threat of deviators being punished in future rounds.  This is now of particular interest in the context of the nascent research area of {\em cooperative AI}~\cite{Dafoe21:Cooperative}, especially for game-theoretic approaches to that~\cite{Conitzer23:Foundations}.
Thanks to the folk theorem, they also allow for more efficient computation of Nash equilibria than one-shot games, as was observed by~\cite{Littman05:Polynomial} for two players. With three or more players the problem becomes hard again~\cite{Borgs10:Myth}, though in practice it can often be solved fast~\cite{Andersen13:Fast} and if correlated punishment is allowed then the problem becomes easy again~\cite{Kontogiannis08:Equilibrium}.

Separately, the role of {\em anonymity} in game theory has long been of interest to the AI community.  Perhaps most significantly, in mechanism design, there is a long line of work on {\em false-name-proof} mechanisms, in which agents cannot benefit from participating multiple times using fake identifiers~\cite{Yokoo01:Robust,Yokoo04:Effect,Conitzer10:Using}.
This is conceptually related to the work in this paper, insofar as an agent that restarts with a different partner makes use of a degree of anonymity in the system.  However, in our context an agent does not use multiple identifiers simultaneously, and the work seems technically quite distinct.

Cooperation in repeated games with the possibility of rematching with a new partner
has been studied in the economic theory community, similarly identifying the importance of building up relations gradually. Our work differs from this literature in that we focus on the computational problem of optimizing the equilibrium, for arbitrary games, whereas the economics literature focuses on obtaining characterization results in specific settings such as lending and borrowing games \cite{Saikat:lender_borrower,WEI2019311}, prisoner's dilemma \cite{Fujiwara-Greve:VoluntaryPD,RobYan:LongTerm,Izquierdo:Leave}, and environments with agents of multiple types \cite{Kranton:repeated,Ghosh:Repeated,RobYan:LongTerm}.

\subsection{Overview}

In Section \ref{sec:prelim}, we introduce the concepts and notation. In Section \ref{sec:optseq}, we introduce optimal sequences and prove their various properties. In Section \ref{sec:eq_class}, we define an equivalence relation on sequences based on their total discounted utility with high discount factor. In Section \ref{sec:comp_prob}, we formalize the computational problem of computing a representative of the optimal equivalence class of this relation. As our main results, we prove an NP-hardness result and present three algorithms for the problem: a pseudo-polynomial time dynamic program and an integer linear program that exactly solves it, as well as a fully polynomial time approximation scheme. In Section \ref{sec:exp}, we report runtimes from our experiments with these algorithms. We present directions for future research in Section \ref{sec:conc}.

\section{Preliminaries}\label{sec:prelim}
In this work, we restrict our attention to \textit{symmetric games and strategies}, which allows us to condense the representation of a game and thereby simplify the presentation. We discuss moving beyond symmetry in Section \ref{sec:conc}.

\subsection{Problem Instance}
 Given a two-player finite, symmetric normal form game $\Gamma$ with actions $A=\{a^{(1)},...,a^{(n)}\}$ and integer payoffs, we define:
\begin{itemize}
    \item The \textit{cooperative payoff} function $p: A \rightarrow \mathbb{Z}$, which maps $a^{(j)}$ to the payoff that the two players receive in $\Gamma$ if they both play $a^{(j)}$ for each $j \in [n]$.
    \item The \textit{deviation payoff} function $p^*: A \rightarrow \mathbb{Z}$, which maps each $a^{(j)}$ to the max.\ payoff a player can achieve given they play any $A \setminus \{a^{(j)}\}$ and their opponent plays $a^{(j)}$.
    \item The \textit{discount factor} $\beta \in [0,1)$, such that if a player receives payoff $p_i$ in round $i \in \mathbb{N}$, her total discounted utility is $\sum_{i=0}^\infty \beta^i p_i$. To avoid confusion, we use subscript $i \in \mathbb{N} = \{0,1,\ldots\}$ to iterate over the rounds, and superscript $j\in [n] =\{1,\ldots, n\}$ to iterate over actions.
\end{itemize}
Thus any finite symmetric game $\Gamma$ can for our purposes be represented as $G=(\{(p^{(j)}, p^{*(j)})\}_{j \in [n]}, \beta)$, where $p^{(j)} \defeq p(a^{(j)})$ and  $p^{*(j)} \defeq p^*(a^{(j)})$. For instance, the game in Table \ref{table:game} is represented as $G_{\ref{table:game}}=(\{(4,0), (5,11),(8,14)\}, \beta)$. 
\subsection{Equilibria with Restarts}
In this paper, we focus on strategy profiles $\boldsymbol{\sigma}$ where each player plans to follow the sequence of moves $(a_0,a_1,...) \in A^\mathbb{N}$, and will restart the sequence if the other player deviates from it. We define $p_i \defeq p(a_i)$ and $p^*_i \defeq p^*(a_i)$ for all $i\in \mathbb{N}$. Two forever-cooperating agents will achieve total discounted utility $\sum_{i=0}^\infty \beta^i p_i$, whereas an agent deviating on round $k$ will achieve total discounted utility $\sum_{i=0}^{k-1} \beta^i p_i+ \beta^k p^*_k + \sum_{i=0}^\infty \beta^{i+k+1} p_i$.\footnote{Note that this assumes the deviating agent never deviates again; this is WLOG, since deviating repeatedly in this way is an improvement if and only if deviating once in this way is an improvement.} Accordingly, for $\boldsymbol{\sigma}$ to be stable (i.e., a Nash equilibrium), we need (for all $k \in \mathbb{N}$): 
\begin{align}
    \sum_{i=0}^{k-1} \beta^i p_i+ \beta^k p^*_k
 \leq (1-\beta^{k+1}) \sum_{i=0}^\infty \beta^{i} p_i \label{eq:stable}
\end{align}
 
Since $\boldsymbol{\sigma}$ is completely defined by the planned sequence, we can succinctly represent it using its corresponding sequence of payoffs $\left(p_i\right)_{i\in \mathbb{N}} \in \{p^{(j)}\}_{j \in [n]}^\mathbb{N}$ and $\left(p^*_i\right)_{i\in \mathbb{N}}\in \{p^{*(j)}\}_{j \in [n]}^\mathbb{N}$.

\section{Optimal Sequences}\label{sec:optseq}
Given a game instance $G$, there can be infinitely many stable sequences $(p_i)_{i\in \mathbb{N}}$. For instance, given $G_{\ref{table:game}}$, our representation of the game in Table \ref{table:game}, both of the following sequences are stable for $\beta=0.9$, as they fulfill Equation (\ref{eq:stable}) for all $k \in \mathbb{N}$:
\begin{align*}
p^a_i = \begin{cases}4 &\text{if }i \leq 10\\ 5 &\text{otherwise}\end{cases} ~~\text{and} ~~ p^b_i = \begin{cases}4 &\text{if }i \leq 3 \\ 8&\text{if }i > 3, i\text{ even} \\ 5 &\text{otherwise}\end{cases} 
\end{align*}
The multitude of stable sequences for a given game raises two questions: how can we (i)  determine which sequences are more desirable than others, and (ii) compute the most desirable sequences? The answer to (i) is straightforward in our context of symmetric games, where the payoffs of two cooperating players are the same. Thus, one can focus on maximizing total discounted utility without fairness concerns. For instance, $\sum_{i=0}^\infty \beta^i p_i^a \approx 43.1 $ and $\sum_{i=0}^\infty \beta^i p_i^b \approx 56.9 $, indicating $(p_i^b)_{i \in \mathbb{N}}$ to be `the better plan'. More generally, we define:
\begin{definition}[Optimal sequence] Given any $G=(\{(p^{(j)},p^{*(j)})\}_{j \in [n]}, \beta)$, a sequence of payoffs $(p_i)_{i \in \mathbb{N}}$ \textbf{surpasses} sequence $(p'_i)_{i \in \mathbb{N}}$ if  $\sum_{i=0}^\infty \beta^i p_i > \sum_{i=0}^\infty \beta^i p'_i $. A sequence $(p_i)_{i \in \mathbb{N}}$ is an \textbf{optimal sequence} if:
\begin{enumerate}
    \item it is stable, i.e., satisfies Equation (\ref{eq:stable}) for all $k \in \mathbb{N}$, and
    \item it is not surpassed by any other other stable sequence.
\end{enumerate}
\end{definition}
\noindent We first prove an existence result:
\begin{proposition} \label{prop:opt_existence}
    For any $G=(\{(p^{(j)},p^{*(j)})\}_{j\in [n]}, \beta)$, if there is any stable sequence, then an optimal sequence exists. This sequence is not necessarily unique.
\end{proposition}
\begin{proof}
    The proof follows from the following claim. To save space, we present the proof of the claim as well as an example for non-uniqueness in Appendix \ref{appendix:lemma1_proofs}.
    
    \begin{claim}\label{claim:closed}
        Given any instance $G$, the set of achieveable total utilities is a closed set, where we say total $t \in \mathbb{R}$ is achievable if there exists a stable sequence $(p_i)_{i\in \mathbb{N}}$ with  $\sum_{i=0}^\infty \beta^i p_i = t$.
    \end{claim}
 Since the set of achievable total discounted utilities is closed and bounded, it contains its supremum \cite[Thm. 2.28]{rudin1976principles}, proving an optimal stable sequence exists.
\end{proof}
Having proven the existence of optimal sequences, we now present two lemmas about their properties.
\begin{lemma} \label{lemma:repeating}
    Any optimal sequence will reach a step after which a single payoff (the highest in the sequence) is repeated forever. We call this payoff the \textbf{goal value} of the sequence.
\end{lemma}
Intuitively, Lemma \ref{lemma:repeating} implies that any optimal sequence consists of an infinitely repeating goal value, the highest payoff, as well as a preceding \textit{hazing period}, a finite sequence of non-maximal payoffs. These two stages are interdependent: an agent cooperates during the hazing period due to the promise of the goal value, whereas any defection after reaching the goal value is avoided by the threat of facing the hazing period again.
\begin{lemma} \label{lemma:optimalgoal}
    For large enough $\beta$, the goal value of any optimal sequence is $p_{\Omega} \defeq \max_{j \in [n]} p^{(j)}$. 
\end{lemma}
Together, Lemmas \ref{lemma:repeating} and \ref{lemma:optimalgoal} (proven in Appendix \ref{appendix:goal_lemmas}) allow easily ruling out many sequences as non-optimal for sufficiently high values of $\beta$, including $(p^a_i)_{i\in\mathbb{N}}$ and  $(p^b_i)_{i\in\mathbb{N}}$ from above, the latter since it never converges to a goal value, and the former since it converges to one that is not the highest in the game. As we have seen in Lemma \ref{lemma:optimalgoal}, the large $\beta$ setting achieves the largest goal values, and accordingly, the largest gap between achieving cooperation and failing to do so, adding to the significance of computing stable and optimal sequences. Indeed, in much of the literature on repeated games in general (including folk theorems), the
focus is on the limit case where $\beta \rightarrow 1$.
Accordingly, in the next section, we study the optimality of sequences in the $\beta \rightarrow 1$ limit. 
\section{Limit-Utility Equivalence Classes}\label{sec:eq_class}
We say that a sequence $(p_i)_{i \in \mathbb{N}}$ is stable in the $\beta \rightarrow 1$ limit if there exists a $\beta'$ such that $(p_i)_{i \in \mathbb{N}}$ is stable for all $\beta>\beta'$. 
Of course, as $\beta\rightarrow 1$, the total discounted utility diverges. Hence, we now introduce an equivalence relation that allows us to compare the total discounted utilities in this limit:
\begin{definition}[Limit-utility equivalence]
    Given a game $G=\{(p^{(j)}, p^{*(j)})\}_{j \in [n]}$, two stable sequences $(p_i)_{i\in \mathbb{N}}$ and $(p'_i)_{i\in \mathbb{N}}$ are \textbf{limit-utility equivalent} if $\lim_{\beta \rightarrow 1} \sum_{i=0}^\infty \beta^i (p_i-p'_i) = 0$.  
\end{definition}
 Indeed, this relationship is symmetric, reflexive, and transitive. We now introduce the optimal equivalence class: 

\begin{proposition} \label{proposition:opteq}There exists a well-defined optimal limit-utility equivalence class: that is, an equivalence class such that for any member sequence $(p_i)_{i\in \mathbb{N}}$ and any other sequence $(p'_i)_{i\in \mathbb{N}}$, we have $\lim_{\beta \rightarrow 1} \sum_{i=0}^\infty \beta^i (p_i-p'_i) \geq 0$.
\end{proposition}
\begin{proof}
    By Lemmas \ref{lemma:repeating} and \ref{lemma:optimalgoal}, any $(p_i)_{i\in \mathbb{N}}$ that converges to $p_\Omega = \max_{j \in [n]} p^{(j)}$ will surpass any $(p'_i)_{i\in \mathbb{N}}$ that does not, implying the equivalence class of the latter cannot be optimal. Ruling such sequences out allows us to define for any sequence of payoffs $(p_i)_{i \in \mathbb{N}}$ a corresponding sequence of hazings $(h_i)_{i \in \mathbb \mathbb{N}}$, where $h_i = p_\Omega - p_i$. For any two sequences that converge to $p_\Omega$, we have $\lim_{\beta \rightarrow 1} \sum_{i=0}^\infty \beta^i (p_i-p'_i) = \sum_{i=0}^\infty h'_i -\sum_{i=0}^\infty h_i$, where neither of the sums diverge since $h_i \neq 0$ or $h'_i \neq 0$ for only finitely many $i$. Hence, the equivalence class of a sequence with goal value $p_\Omega$ is entirely determined by the (non-discounted) sum of its per-round hazings. By Lemma \ref{lemma:optimalgoal}, there exists at least one such stable sequence $(h_i)_{i \in \mathbb{N}}$, say with total hazing $H= \sum_{i=0}^\infty h_i$. Since the total non-discounted hazing for any sequence is an integer and bounded below by 0, there can only be finitely many improvements over $H$, implying that there is a well-defined minimum total hazing $H_{min}$ with at least one corresponding stable sequence. Hence, any stable sequence with total hazing $H_{min}$ cannot be surpassed (as $\beta \rightarrow 1$) by any other stable sequence, either converging to $p_\Omega$ or not, proving its equivalence class is optimal. 
\end{proof}
\section{Computational Problem: OptRep}\label{sec:comp_prob} 
Having proven the existence of the optimal equivalence class, the natural next question becomes: how do we compute a (representative) sequence from this class? Before formalizing this as computational problem, we investigate the stability condition given in Equation (\ref{eq:stable}) in the $\beta \rightarrow 1$ limit:

\begin{proposition}\label{prop:hstability}
    A sequence $(p_i)_{i \in \mathbb{N}}$ with goal value $p_\Omega = \max_{j \in [n]} p^{(j)}$ is stable in the $\beta \rightarrow 1$ limit if and only if the following holds for all $k \in \mathbb{N}:$
\end{proposition}
\begin{align}
    \sum_{i=0}^{k-1} (p_\Omega - p_i) > p^*_k - p_\Omega \label{eq:limitstability}
\end{align}
The proof is given in Appendix \ref{appendix:hstability}. Intuitively, the left-hand side of (\ref{eq:limitstability}) is the cost of restarting, as repeating each step $i$ results in a loss of $p_\Omega-p_i$ (the so-called hazing cost) compared to the goal value, whereas the right-hand side is the gains from deviating, represented by the advantage over the goal value one can gain by deviating now. The inequality needs to be strict because, intuitively, for any $\beta<1$, a ``tie'' between the two sides of the equation would be broken towards the side of deviation, as that deviation payoff comes earlier.

Motivated by this way of expressing the stability condition, we modify the representation of a game to fit our problem: Given $G=\{(p^{(j)},p^{*(j)})\}_{j \in [n]}$ and $p_\Omega= \max_{j\in[n]} p^{(j)}$, for each $j \in [n]$ we define a corresponding \textbf{\textit{hazing cost}} $h^{(j)} = p_\Omega - p^{(j)}$ and a \textbf{\textit{threshold}} $t^{(j)} = p^{*(j)}-p_\Omega$. We now represent the game as $G=\{(h^{(j)}, t^{(j)})\}_{j \in [n]}$ and any sequence of payoffs $(p_i)_{i \in \mathbb{N}}$ with the corresponding sequence of hazing costs $(h_i)_{i \in \mathbb{N}}$. The stability condition from (\ref{eq:limitstability}) then becomes:
\begin{align}
    \sum_{i=0}^{k-1}h_i > t_k \label{eq:hstability}
\end{align}
As shown in the proof of Prop \ref{proposition:opteq}, all members of the optimal limit-utility equivalence class necessarily have $p_\Omega$ as their goal value, and hence $h_i \neq 0$ for finitely many $i$, and the equivalence class of any such sequences is entirely determined by $\sum_{i=0}^\infty h_i$. Any such stable sequence must satisfy $\sum_{i=0}^\infty h_i > p^*_\Omega - p_\Omega \equiv \Delta$, to ensure (\ref{eq:hstability}) is fulfilled for the steps after reaching the goal value. Thus, the goal of finding a representative of the optimal equivalence class reduces to:

\begin{definition}[OptRep] Given $\{(h^{(j)}, t^{(j)})\}_{j \in [n]} \subset  \mathbb{Z}_+ \times \mathbb{Z}$ and $\Delta \in \mathbb{Z}_+$, \textbf{OptRep} asks to find a finite sequence $\left(h_i \right)_{i \in \{0\}\cup [\ell]} \in \{h^{(j)}\}_{j \in [n]}^{\ell+1}$ for some $\ell \in \mathbb{N}$ such that $\sum_{i=0}^{\ell}h_i$ is minimized, subject to:
\begin{center}
    $(\forall k \in \{0\}\cup[\ell]): \sum_{i=0}^{k-1} h_i > t_k
    \textbf{ and } \sum_{i=0}^{\ell} h_i >\Delta$
\end{center}
We now prove a hardness result for OptRep:
    
\end{definition}
\begin{theorem} \label{thm:nphard}
    OptRep is (weakly) NP-hard, and the corresponding decision problem is NP-complete. 
\end{theorem}
\begin{proof}

 We will prove the theorem by reducing the Unbounded Subset-Sum Problem (USSP) to OptRep. USSP is NP-complete \cite{Lueker:TwoNP} and asks:
 
 Given positive integers $\left\{ a_{i}\right\}_{i \in [m]}$ and $A$, do there exist non-negative integers $\left\{ r_{i}\right\}_{i \in [m]}$ such that
$
\sum_{i=1}^m r_i \cdot a_{i}=A ?$

We now present our reduction:
      For each $j\in [m]$, define $h^{(j)}= a_j$ and $t^{(j)} = -1$ (in the payoff representation, this simply implies $p^{*(j)}= p_\Omega-1$). Lastly, set $\Delta=A-1$. We claim that the answer to the Unbounded Subset-Sum Problem is yes if and only if the total hazing of the optimal sequence computed by OptRep on input $G=(\{(h^{(j)},t^{(j)})\}_{j \in [m]}, \Delta)$ is $A$. The backwards direction is obvious. For the forward direction, say that the answer to USSP is yes, with coefficents $(r_j)_{j \in [m]}$. We can construct a finite hazing sequence $(h_i)_{i \in \{0\} \cup [R]}$ that has each $h^{(j)}$ repeating $r_j$ times, where $R= \sum_{j=1}^m r_j$. Notice that (\ref{eq:hstability}) is fulfilled since $t_i = -1$ for all $i \in \{0\} \cup [R]$. As for the final hazing threshold, we have $\sum_{i=0}^R h_i=  \sum_{j=1}^m r_j a_j  = A= \Delta+1>\Delta$, so the overall sequence is stable. Since the total hazing will be integral, any sequence with less total hazing will fail to meet the threshold set by $\Delta$. Hence, the total hazing of the optimal sequence computed by OptRep on input $G=(\{(h^{(j)},t^{(j)})\}_{j \in m}, \Delta)$ is $A$. 
  
    The reduction proves that OptRep is NP-hard. The corresponding decision problem (of whether a stable sequence with total hazing $H$ in the $\beta \rightarrow 1$ limit exists) is NP-complete, since, as we will see in Lemma \ref{lemma:threshold_monotonicity}, every sequence has a polynomial-size representation that does not change its total hazing cost, which can be computed in polynomial time. 
 \end{proof}
As shown by \cite[Theorem 4]{rationalNPhard}, USSP becomes \textit{strongly} NP-complete when the inputs are rational numbers rather than integers. Due to our reduction above, the same result readily translates to OptRep:
\begin{corollary} OptRep with rational inputs is strongly NP-hard, and the corresponding decision problem is strongly NP-complete.
\end{corollary}

Considering the similarities between OptRep and USSP presented in the proof of Theorem \ref{thm:nphard}, one might wonder if a reduction from OptRep to USSP is also possible. However, OptRep has additional challenges: instead of a set of items, we need to pick a sequence and ensure that every step of it meets the requirement in (\ref{eq:hstability}), as opposed to only a final capacity requirement. 

 Despite these challenges, we can exploit the structural properties of OptRep to develop techniques similar to those employed for knapsack problem variants such as USSP. One key observation is that the stability of a given sequence at any step solely depends on the hazing that has already occurred. This enables us to solve the problem with a dynamic program with a single state variable (the hazing so far), leading to:
 
 \begin{theorem} \label{thm:pseudo}
OptRep is solvable in pseudo-polynomial time. \end{theorem}
 \begin{proof}
     We prove the theorem by showing the correctness, time complexity, and space complexity of Algorithm \ref{alg:dp}.

     \noindent \textit{(a) Correctness:} Due to integral hazing costs, the minimum hazing cost for any stable sequence is $\Delta+1$. We prove (by strong induction on decreasing $h$) that once Algorithm \ref{alg:dp} is complete, $D[h]$ contains the minimum overshoot over this lower bound one can achieve starting from a hazing cost of $h$. As a base case, for any $h \geq \Delta +1$, the hazing goal has been met and the best we can do is an overshoot of $h-\Delta-1$ (Line \ref{line:basecase}). For any $h<\Delta+1$, at least one more action needs to be added to meet the hazing goal. In Line \ref{line:choose}, the algorithm chooses the next eligible action that minimizes the overshoot, where action $j\in[n]$ is eligible if the hazing so far, $h$, has met its threshold, $t^{(j)}$. Inductively assuming that $D[h']$ is correct for all $h'>h$ ensures that $D[h]$ is correctly set. Accordingly, $D[0]$ gives the overshoot of a sequence from the optimal equivalence class, and the list $A$ keeps track of the next steps to later reconstruct the sequence.

     \noindent \textit{(b) Time complexity:} Line \ref{line:basecase} takes constant time and is implemented $h_{\max}=\max_{j \in [n]}h^{(j)}$ times. Lines \ref{line:induc_start}-\ref{line:induc_end} take $O(n)$ time and are executed $O(\Delta)$ times. Since the sequence returned has at most $\Delta+1$ actions, Lines \ref{line:while_start}-\ref{line:while_end} take $O(\Delta)$ time. All other steps take constant time. Overall, the algorithm runs in $O(n\Delta+h_{max})$ time, which can be improved to $O(n\Delta)$ by feeding Line \ref{line:basecase} into the definition of $D[h]$ in Lines  \ref{line:induc_start}-\ref{line:induc_end}. This is polynomial in the input but, since input is represented in binary, exponential in the problem size.

     \noindent \textit{(c) Space complexity:} $D$ and $A$ both have $\Delta$ variable entries, giving the algorithm a space complexity of $O(\Delta)$. If we are not interested in reconstructing the representative sequence, the space complexity improves to $O(\min\{h_{max},\Delta\})$, since we may only store the most recent $h_{max}$ entries of $D$, which are sufficient for Line \ref{line:choose}. 
     \end{proof}
 
\begin{algorithm}[tb]
    \caption{Dynamic Program for OptRep}
    \label{alg:dp}
    \textbf{Input}: Final hazing goal: $\Delta>0 ,$ tuples of (hazing, thresholds) values for each action: $ \{(h^{(j)}, t^{(j)})\}_{j\in [n]} \subset \mathbb{Z}_+ \times \mathbb{Z}$\\
    \textbf{Output}: a list $S_{o}$ representing a sequence from the optimal equivalence class, along its hazing cost $H_o \in \mathbb{Z}_+$.
    
    \begin{algorithmic}[1] %[1] enables line numbers
        \State Let $\ell= \Delta + \max_{j\in[n]} h^{(j)}$. %\COMMENT{Length of the lookup table}
        \State $D \gets[\textit{inf}]^\ell$, \quad $A \gets[\textit{inf}]^\ell$ %\Comment{Indices of the optimal next action}
        \For{$h = \ell, \ell-1, \ldots, 1,0$}
        \InlineIf{$h \geq \Delta+1$}{ $D[h] \gets h-\Delta-1$} \label{line:basecase}
        \State \textbf{else} \label{line:induc_start}
       \State  \hspace{1em} $S_h \gets \{j \in [n] : h> t^{(j)} \}$ 
        \State \hspace{1em} $D[h] \gets \min_{j \in S_h} D[h+h^{(j)}]$  \label{line:choose}
        \State \hspace{1em} $A[h] \gets \argmin_{j \in S_h} D[h+h^{(j)}]$  \label{line:induc_end}
        \EndFor
        \State $H_o \gets D[0]+\Delta +1$,\quad $S_o \gets []$, \quad $i \gets 0$
        \While{$i \neq \textit{inf}$}\label{line:while_start}
        \State $S_o \gets S_o + [A[i]]$, \quad $i \gets A[i]$
        \EndWhile\label{line:while_end}
        \State \textbf{return} $S_o,H_o$
    \end{algorithmic}
\end{algorithm}

Alternatively, one can formulate OptRep as an Integer Linear Program (ILP), using binary variables for whether action $(j)$ is the $i$th action in the sequence. But the hazing period can have as many as $\Delta+1$ actions, resulting in $O(n\Delta)$ variables, exponentially many. Instead, we ask if one can impose some structure on the output sequence without losing generality.

One intuitive candidate is monotonicity: does the optimal equivalence class always have a sequence with non-increasing hazing costs (non-decreasing payoffs), considering the goal value is the highest payoff? The answer, it turns out, is no. Consider: $\Delta=10$ with $(h^{(1)},t^{(1)})=(5,-1),(h^{(2)},t^{(2)})=(6,4)$. The sequence $(h^{(1)},h^{(2)})$ is stable and achieves a hazing of 11. The minimum hazing from any non-increasing sequence is 15, by $(h^{(1)},h^{(1)},h^{(1)})$.

There is, however, a property that can be imposed on the sequences without ruling out the optimal equivalence class:

\begin{lemma}\label{lemma:threshold_monotonicity}
     Any stable sequence of the optimal equivalence class can be converted to a stable sequence in the same class where (a) all appearances of any action are adjacent and (b) actions appear in order $j_1,j_2,...,j_\ell$ where $t^{(j_1)} \leq t^{(j_2)} \leq \ldots \leq t^{(j_\ell)}$. We call such sequences $\textbf{threshold-monotonic}.$
\end{lemma}
\begin{proof}
    For a given action $j$, say $i_j$ is its first appearance in the sequence. Move all appearances of $h^{(j)}$ to appear immediately after $i_j$, pushing all the actions that were previously there to later steps, without changing their order. The total hazing (and hence the equivalence class) has not changed. The stability of all $h^{(j)}$ follows from its stability at step $i_j$: once its threshold is met, this action will always be stable. All other actions are still stable since the hazing that precede them could have only increased. Repeating this for every $j \in [n]$ achieves (a). 

    Say the actions now appear in order $j_{1},j_2,\ldots, j_{n}$. Assume the first $k$ of them are the lowest $k$ threshold actions (with the base case $k=0$). Then the action with $(k+1)^\text{th}$ lowest threshold is $j_{z}$ for some $z \in \{k+1,\ldots n\}$. Move every occurrence of action $j_{z}$ to the starting point of action $j_{k+1}$, sliding every other action to the right. The occurrences of action $j_{z}$ will be stable by the stability of $j_{k+1}$ prior to the shift, which has an equal or greater threshold. All other actions are stable as their preceding hazing can only increase. Repeating this  inductively results in the overall sequence fulfilling (b). 
\end{proof}
 
 By Lemma \ref{lemma:threshold_monotonicity}, we can restrict our attention to threshold-monotonic sequences while solving OptRep, leading to an ILP with $O(n)$ variables, presented in Algorithm \ref{alg:ilp}. Note that ordering actions by thresholds takes an additional $O(n\log n)$ preprocessing time. The assumption $t^{(n)}\leq \Delta$ is WLOG, since any action $j$ with $t^{(j)} \geq \Delta +1$ cannot be used until the final hazing threshold is met. The ILP returns $[r^{(j)}]_{j \in [n]}$, the number of times each action is repeated, as this is sufficient to represent a threshold-monotonic sequence. As seen in our experiments in Section \ref{sec:exp},  Algo. \ref{alg:ilp} is efficient in practice.

\begin{algorithm}[tb]
    \caption{Integer Linear Program for OptRep}
    \label{alg:ilp}
    \textbf{Input}: Final hazing threshold: $\Delta>0 ,$ tuples of (hazing, thresholds) values for each action: $ \{(h^{(j)}, t^{(j)})\}_{j\in [n]} \subset \mathbb{Z}_+ \times \mathbb{Z}$, with $t^{(1)}\leq t^{(2)} \leq \ldots \leq t^{(n)}  \leq \Delta$\\
    \textbf{Parameters}: $(\forall j \in [n]): r^{(j)}$ indicating how many times action $j$ is repeated in the output sequence\\
    \textbf{Output}: Final values of $[r^{(j)}]_{j \in [n]}$
    \begin{algorithmic}[1] %[1] enables line numbers
        \State \textbf{if }{$t^{(1)} \geq 0$} \textbf{then raise} \textit{error} 
        \State \textbf{minimize} $\sum_{j=1}^n r^{(j)} h^{(j)}$
        \State \textbf{subject to}
        \State \quad $(\forall j \in [n] \setminus\{1\})$ $\sum_{j'=1}^{j-1} r^{(j')}h^{(j')} \geq t^{(j)}+1$
        \State \quad $\sum_{j=1}^{n} r^{(j)}h^{(j)} \geq \Delta+1$
        \State \textbf{return} $[r^{(j)}]_{j \in [n]}$
    \end{algorithmic}
\end{algorithm}

  Lemma \ref{lemma:threshold_monotonicity} can also improve the space complexity of Algorithm \ref{alg:dp}: each $A[h]$ can now store a size-$n$ vector corresponding to $[r_j]_{j\in[n]}$, with Line \ref{line:induc_end} replaced with $A[h][j^*] \leftarrow A[h][j^*]+1$ where $j^* = \argmin_{j\in S_h} D[h+h^{(j)}]$. Thus, it now suffices to store only $h_{max}$ entries of $A$ and $D$. Call Algorithm 1 with this modification (as well as those mentioned in the proof of Theorem \ref{thm:pseudo}(b-c)) Algorithm $1^*$. This leads to the below strengthening of Theorem \ref{thm:pseudo}:
 \begin{theorem}
     Algorithm $1^*$ solves OptRep in $O(n\Delta)$ time using $O(\min\{\Delta, n\cdot h_{max} \})$ space. The space complexity is $O(\min\{\Delta,  h_{max} \})$ for the corresponding decision problem. 
 \end{theorem}
While efficient in practice, both Algorithms \ref{alg:dp} and \ref{alg:ilp} are exponential in the worst case, which motivates the question of whether OptRep can be approximated with a Fully Polynomial Time Approximation Scheme (FPTAS), which runs in polynomial in the problem size and in $1/\varepsilon$, where $\varepsilon$ is the approximation ratio, given as an input. 
While we have  given a pseudopolynomial-time algorithm for our problem, doing so in general does not guarantee the existence of an FPTAS: for instance, knapsack with multiple constraints is pseudo-polynomial-time solvable but does not have an FPTAS unless $P=NP$ \cite{Magazine:MultiFPTAS}. 
However, for \mbox{OptRep} there is in fact an FPTAS, given as Algorithm \ref{alg:fptas}. The algorithm is a modification of Ibarra and Kim's original FPTAS for unbounded knapsack \shortcite{ibarra_kim_FPTAS}, with additions that address the differences between OptRep and unbounded knapsack.

\begin{figure*}[t]
  \includegraphics[scale=0.205]{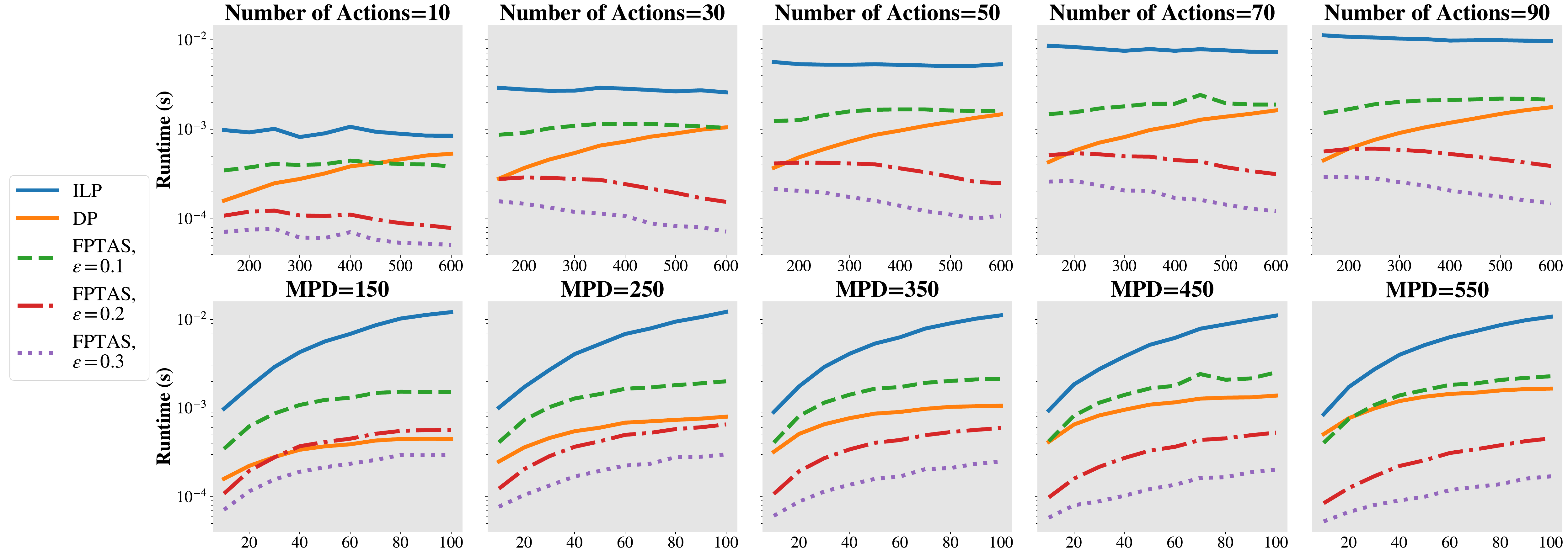}
  \caption{The semi-log plots of runtimes of Algorithms \ref{alg:dp},\ref{alg:ilp}, and \ref{alg:fptas} (with $\varepsilon=0.3,0.2,$ and $0.1$). Each data point is averaged over 5000 trials. \textbf{Top Row:} Fixed number of actions $(n)$; $x$-axis shows Maximum Payoff from Deviation (MPD). \textbf{Bottom Row:} Fixed MPD; $x$-axis shows $n$. }
  \label{fig:runtimes}
\end{figure*}

\begin{theorem}\label{thm:fptas}
    Algorithm \ref{alg:fptas} is an FPTAS for OptRep, and runs in $O(n\log n + \frac{n}{\varepsilon^2})$ time.
\end{theorem}
\begin{algorithm}[!t]
    \caption{FPTAS for OptRep}
    \label{alg:fptas}
    \textbf{Input}: Final hazing threshold: $\Delta>0 ,$ tuples of (hazing, thresholds) values for each action: $ \{(h^{(j)}, t^{(j)})\}_{j\in [n]} \subset \mathbb{Z}_+ \times \mathbb{Z}$, with $t^{(1)}\leq t^{(2)} \leq \ldots \leq t^{(n)}  \leq \Delta$\\
    \textbf{Output}: $\hat{H}, \hat{L}$, where $\hat{L}$ is a list of actions representing the output sequence, and $\hat{H}$ is the corresponding hazing cost. 
    
    \begin{algorithmic}[1] %[1] enables line numbers
        \State \textbf{if }{$t^{(1)} \geq 0$} \textbf{then raise} \textit{error}\label{line:error}
        \State $F \gets \{j \in [n]: t_i < 0, h_i \leq \Delta \}$ \label{line:eligible}
        \If {$F = \emptyset$}
        \State $\hat{L} \gets [\argmin_{j \in [n]: t_i < 0} h^{(j)}]$, \quad $\hat{H} \gets h^{(\hat{L}[0])}$\label{line:pre_no_halt}
        \State \textbf{return} $\hat{H}, \hat{L}$ \label{line:no_halt}
        \EndIf 
        \State $j^* \gets \min F$, \quad  $\Tilde{H}\gets k\cdot h^{(j^*)}$ where $k$ is the min. integer s.t. $k\cdot h^{(j^*) } > \Delta$ \label{line:tilde}
        \State $\delta \gets \Tilde{H}\left(\frac{\varepsilon}{3}\right)^2$, $g \gets \left\lfloor \frac{\Tilde{H}}{\delta} \right\rfloor =\left\lfloor \left( \frac{3}{\varepsilon} \right)^2 \right\rfloor $, $S \gets \emptyset$\label{line:normal_fact} 
        \State $T\gets [\textit{inf}]^{g+1}$, $L\gets [\textit{inf}]^{g+1}$, $T[0] \gets 0 $, $L[0] \gets [~]$ \label{line:prefor}
        \For{$j = 1,\ldots,n$} \label{line:outer}
        \InlineIf{$h^{(j)} \leq \frac{\varepsilon}{3} \Tilde{H}$}{$S \gets S + [j]$}\label{line:small}
        \State \textbf{else}
        \State  \hspace{1em} $f^{(j)} \gets \left \lfloor \frac{h^{(j)}}{\delta} \right \rfloor$ \label{line:normal_hazing}
        \State  \hspace{1em}\textbf{for }{$k=0,...,g-f^{(j)}$}\textbf{ do} \label{line:inner}
        \State \hspace{2em}\begin{minipage}[t]{19em}
         \textbf{if} $T[k] \neq \textit{inf}$ \textbf{and} $T[k] > t^{(j)}$ \textbf{and }($T[k+f^{(j)}]=\textit{inf}$ \textbf{or} $T[k]+h^{(j)}>T[k+f^{(j)}]$) \textbf{then}
               \end{minipage} \label{line:substability}
        \State\hspace{3em} $T[k+f^{(j)}] \gets  T[k]+h^{(j)}$\label{line:pre_fptas_endfor}
        \State \hspace{3em} $L[k+ f^{(j)}] \gets L[k]+[j] $\label{line:fptas_endfor}
        \EndFor
        \State $j_s \gets \argmin_{j \in S} t^{(j)}$
        \State $k^* \gets \argmin_{k\in \{0\} \cup [g]: T[k] > t^{(j_{s})}  } T[k]+ n_k\cdot  h^{(j_s)}$ where $n_k $ is the min. integer s.t. $T[k]+ n_k\cdot  h^{(j_s)} > \Delta$ \label{line:finstability}
        \State $\hat{H} \gets T[k^*]+n_{k^*}\cdot h^{(j_s)}$
        , \quad $\hat{L} \gets L[k^*]+ [j_s]^{n_{k^*}}$ \label{line:final_output}
        \State \textbf{return} $\hat{H}, \hat{L}$
        
    \end{algorithmic}
\end{algorithm}

\begin{proof} \textit{(a) Correctness:} To prove correctness, we first present several claims, the proofs of which are in Appendix \ref{appendix:thm3_proofs}

     \begin{claim} \label{claim:tilde}
         Say $H^*$ is the optimal hazing. If the algo. returns on Line \ref{line:no_halt}, $\hat{H}=H^*$. Else, $\Tilde{H} \geq H^* \geq \frac{1}{2} \Tilde{H}$ after Line \ref{line:tilde}.
     \end{claim}   
     Note that Line \ref{line:normal_fact} sets a `normalizing factor' of $\delta = \Tilde{H} \left(\frac{\varepsilon}{3}\right)^2$, later used for computing normalized hazing costs $f^{(j)} = \left\lfloor {h^{(j)}}/{\delta} \right\rfloor$ and a normalized upper bound for the optimal hazing $g= \lfloor {\Tilde{H}}/{\delta} \rfloor$.
     \begin{claim} \label{claim:bounds}
         Any $f^{(j)}$ set on Line \ref{line:normal_hazing} satisfies $f^{(j)} \cdot \delta \leq h^{(j)} \leq f^{(j)} \cdot \delta  \cdot \left(1+\frac{\varepsilon}{2}\right)$ 
     \end{claim}
\begin{claim} \label{claim:agreement}
     At any point during the execution starting from Line \ref{line:prefor}, for any $k \in \{0,...,g\}$, either $T[k]=L[k]=\textit{inf}$ or $L[k]$ contains a list $[j_0,j_1,\ldots, j_{\ell}]$ with $\sum_{i=0}^\ell f^{(j_i)}=k$ and  $\sum_{i=0}^\ell h^{(j_i)}=T[k]$.
\end{claim}
\begin{claim} \label{claim:overshoot}
    Say $(h_i)_{i\in\{0\}\cup [\ell]}$ is any stable threshold-monotonic \textit{sub}sequence (i.e., it fulfills (\ref{eq:hstability}) for all $k \leq \ell$, but does not necessarily meet the final hazing goal $\Delta$), where the order of actions is consistent with the labeling, with $h_i > \frac{\varepsilon}{3}\Tilde{H}$ for all $i$ and total hazing at most $\Tilde{H}$. Defining $f_i = \left \lfloor {h_i}/{g} \right \rfloor$ and $F= \sum_{i=0}^\ell f_i$, we must have $T[F] \geq \sum_{i=0}^\ell h_i$ in the end of the execution.
\end{claim}
We now complete the proof of the theorem: Say $\left(h^\star_i\right)_{i \in \{0\}\cup[\ell]}$ is a member of the optimal equivalence class ($\sum_{i=0}^\ell h^\star_i = H^*)$, with corresponding normalized hazing costs $\left(f^\star_i\right)_{i \in \{0\} \cup [\ell]}$ and $F^* =\sum_{i=0}^\ell f^\star_i$. By Lemma \ref{lemma:threshold_monotonicity}, we can assume WLOG that $\left(h^\star_i\right)_{i \in \{0\} \cup [\ell]}$ is threshold-monotonic, with respect to the current labeling of the actions. 

Say $\hat{H}=T[k^*] +n_{k^*}\cdot h^{j_s}$ and $\hat{L}=L[k^*] +[j_s]^{n_{k^*}}$ are the final outputs of the algorithm, as set in Line \ref{line:final_output}, and $S$ is the set of actions $j$ with $h^{(j)} \leq ({\varepsilon}/{3})\Tilde{H}$, as filled in Line \ref{line:small}. By Claim \ref{claim:agreement}, the total hazing of the sequence represented by $\hat{L}$ is indeed $\hat{H}$. The stability of every step in $L[k^*]$ in ensured by the if statement in Line \ref{line:substability}. The stability of the final $n_{k^*}$ steps is ensured by the $T[k^*]>t^{(j_s)}$ condition in Line \ref{line:finstability}.

\noindent \textbf{Case 1: } $\left(h^\star_i\right)_{i \in \{0\}\cup[\ell]}$ has no element from $S$. Then it fulfills the assumptions of Claim \ref{claim:overshoot} (as $H^* \leq \Tilde{H}$ by Claim \ref{claim:tilde}),  implying $T \left[F^*\right] \geq H^* \geq \Delta + 1 > t^{(j_s)}$. Since $H^* \leq \Tilde{H}$, we have $F^* \leq g$, so $T[F^*]$ is one of the options considered in Line \ref{line:finstability} (with $n_{F^*}=0)$, implying $\hat{H} \leq T[F^*]$. By Claim \ref{claim:agreement}, we get $\sum_{j \in L[F^*]}f^{(j)} = F^* = \sum_{i=0}^\ell f^\star_i$. Using Claim \ref{claim:bounds} twice:
\begin{align*}
    T[F^*]  =\sum_{j \in L[F^*]}h^{(j)} &\leq  \left(1+\frac{\varepsilon}{2}\right)\delta F^* \\&\leq   \left(1+\frac{\varepsilon}{2}\right)\left( \sum_{i=0}^{\ell}h^\star_i\right)
\end{align*}
Hence, $\hat{H} \leq  T[F^*] \leq \left(1+\frac{\varepsilon}{2}\right) H^*$, fulfilling the bound.

\noindent \textbf{Case 2: } $\left(h^\star_i\right)_{i \in \{0\} \cup [\ell]}$ does have an action from $S$.  Say $i^*$ is the first index where such an action appears and $F' = \sum_{i=0}^{i^*-1}f^\star_i$. Then $\left(h^\star_i\right)_{i \in\{0\} \cup [i^*-1]}$ satisfies the assumptions of Claim \ref{claim:overshoot}, so we have $T \left[F'\right] \geq \sum_{i=0}^{i^*-1}h^\star_i > t^\star_{i^*} \geq t^{(j_s)}$, where the inequalities follow from  Claim \ref{claim:overshoot}, stability of $h_{i^*}$ and the minimality of $t^{(j_s)}$, respectively. Then $T[F']+n_{F'} h^{(j_s)}$ is considered in Line  \ref{line:finstability}, implying $\hat{H} \leq T[F']+n_{F'} h^{(j_s)}$. Say $n_{F'}\neq 0$: since $n_{F'}$ is the smallest integer such that $T[F'] + n_{F'} h^{(j_s)}> \Delta$, we have:
\begin{align*}
    \hat{H} &\leq T[F']+(n_{F'}-1)h^{(j_s)} + h^{(j_s)} \leq \Delta + h^{(j_s)}\\& < H^* + \frac{\varepsilon}{3} \Tilde{H} \leq H^* + \frac{2\varepsilon}{3} H^* \leq (1+\varepsilon)H^*
\end{align*}
fulfilling the bound. Otherwise, say $n_{F'}=0$, implying $T[F']>\Delta$. By the same argument as Case 1, using Claims \ref{claim:agreement} and \ref{claim:bounds}, we have $T[F'] \leq \left(1+\frac{\epsilon}{2}\right) \left( \sum_{i=0}^{i^*-1}h_i\right)$. However, we have $ \sum_{i=0}^{i^*-1}h_i \leq \Delta < H^*$ as otherwise item $i^*$ would be unnecessary to reach the goal value. This gives us:
$ \hat{H} \leq T[F'] \leq \left(1+\frac{\varepsilon}{2}\right) H^*$,
fulfilling the bounds. 

Overall, we have shown that in all cases, the output $\hat{L}$ contains a stable sequence with total hazing cost $\hat{H}< (1+\varepsilon) H^*$.

\noindent (b) \textit{Time complexity:} Ordering actions by thresholds takes $O(n \log n)$ time. Lines \ref{line:error}-\ref{line:prefor} take $O(n)$ time. The outer loop in Line \ref{line:outer} gets executed $O(n)$ times, each running a constant number of operations plus the inner loop in Line \ref{line:inner}, which gets executed at most $g=O(1/\varepsilon^2)$ times with constant time per iteration. Thus, Lines \ref{line:outer}-\ref{line:fptas_endfor} take $O\left(\frac{n}{\varepsilon^2}\right)$ time. Lines 30, 31 take $O(n)$ and $O(g)$ time, respectively. Overall, Algo. \ref{alg:fptas} takes $O\left(n\log n + \frac{n}{\varepsilon^2}\right)$ time, polynomial in both $n$ and $\frac{1}{\varepsilon}$. 
\end{proof}
\section{Experiments} \label{sec:exp}
We present the semi-log plots for the runtimes of Algorithms \ref{alg:dp}-\ref{alg:fptas} in Figure \ref{fig:runtimes}. Given the number of actions $n$ and a positive integer Maximum Payoff of Deviation ($\text{MPD}$), we generated a game by (for each action $j\in[n]$) uniformly choosing  $p^{(j)}$ from $\{i \in \mathbb{Z}: 0 \leq i \leq 30 \}$ and uniformly choosing  $p^{*(j)}$ from $\{i \in \mathbb{Z}: p^{(j)} \leq i \leq \text{MPD}\}$.\footnote{Imposing  $p^{*(j)} \geq p^{(j)}$ ensures $\Delta >0$ for each game.} As expected, the runtimes of all algorithms increase with increasing $n$, and the runtime of FPTAS increases with decreasing $\varepsilon$. Similarly, consistent with our runtime analysis, Algorithm \ref{alg:dp} (DP) increases with increasing $\text{MPD}$, since $\text{MPD}$ and $\Delta$ are positively correlated, whereas the other algorithms do not exhibit such an increase, as their runtimes are independent of the actual payoffs.\footnote{In fact, as seen in Fig.\ref{fig:runtimes}(top), the runtime for Algo. \ref{alg:fptas} sometimes decreases with increasing $\text{MPD}$. This is likely because increasing $\text{MPD}$ increases the upper bound for the deviation payoffs but not for the cooperation payoffs, resulting in $\Delta$ increasing while $h^{(j)}$s stay in the same range. Thus, more actions are placed in $S$ and skip the inner loop in Line \ref{line:inner}, resulting in less computation by the FPTAS.} While the DP consistently outperforms the ILP in the plots, this trend is naturally reversed for high $\Delta$, which we demonstrate by rerunning the same runtime experiment for these algorithms for larger values of MPD. Our results are presented in Figure \ref{fig:highdelta}.
    
Notice that for $n=30$, we see that the ILP (the runtime of which is robust to increases in the payoffs) starts to outperform the DP (the runtime of which steadily increases with increasing MPD) when the MPD reaches approximately 1500. The figure demonstrates the better performance of the ILP in the high $\Delta$ limit, which corresponds to games requiring longer (in terms of the number of rounds) hazing periods before the goal value is reached. 

    \begin{figure}[t]
        \centering
        \includegraphics[scale=0.2]{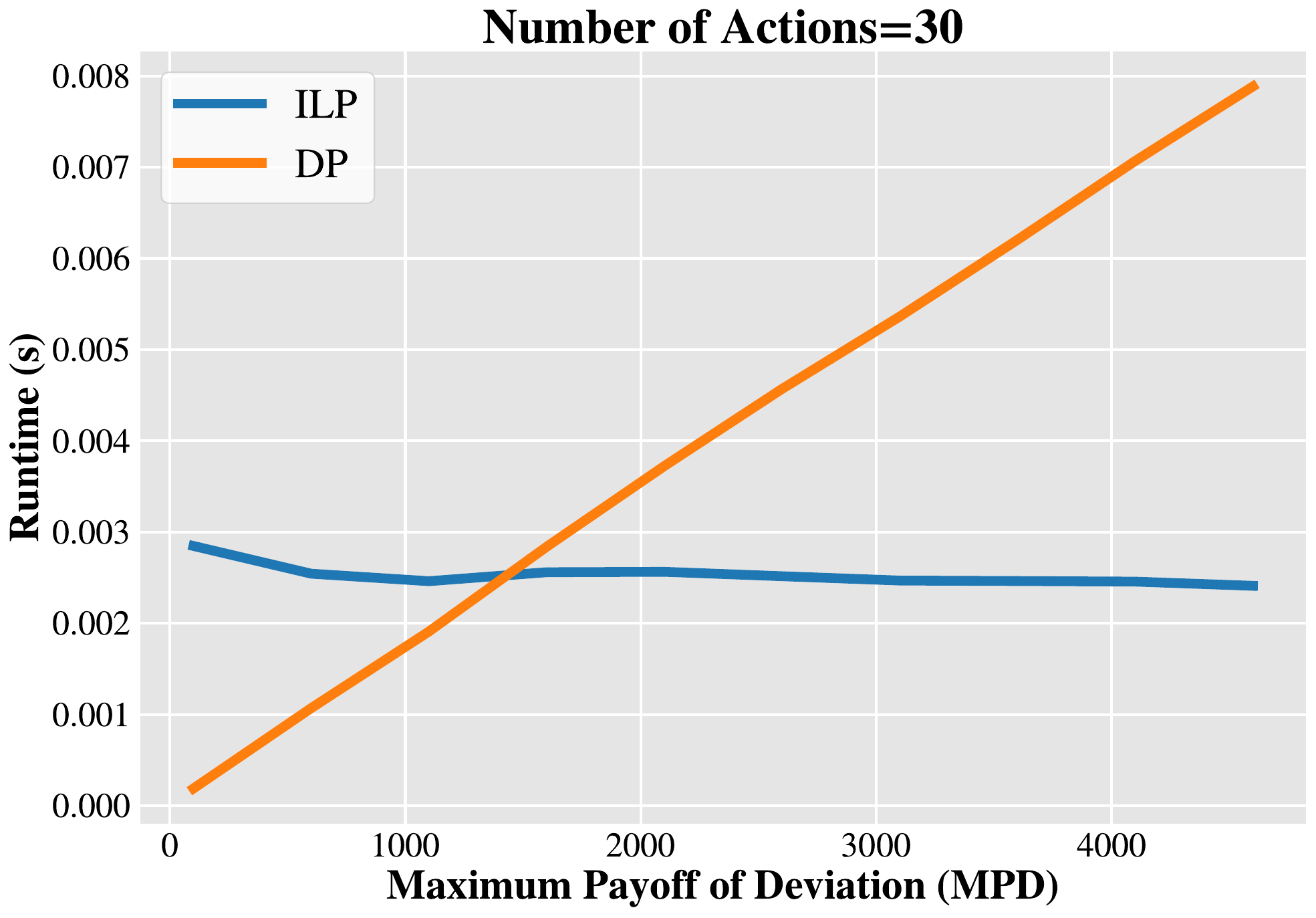}
        \caption{Runtimes of Algorithms \ref{alg:dp} and \ref{alg:ilp} for a fixed number of actions (30) as a function of Maximum Payoff of Deviation (MDP). Each data point is averaged over 5000 trials.}
        \label{fig:highdelta}
    \end{figure}

\section{Future Work} \label{sec:conc}
As mentioned in Section \ref{sec:comp_prob}, Algorithm \ref{alg:fptas} is inspired by the original FPTAS for the unbounded knapsack problem \cite{ibarra_kim_FPTAS}. Since then, faster schemes for the same problem have been developed \cite{Lawler:FPTAS,Jansen:FPTAS}. It is possible that similar modifications can be made to these algorithms to enable them to solve OptRep, improving the runtime of its FPTAS.  

A natural next step following our work is considering asymmetric games and strategies. While some of our results generalize to this setting, there are important differences. For example, goal values must be generalized to goal sequences to maximize payoff (e.g., if $(C_2,C_2)$ from Table \ref{table:game} instead had payoff (6,6), it would be better to forever alternate between $(D,C_2)$ and $(C_2,D)$); hence, computing the “optimal goal sequence” is a separate computational problem on its own. A two-variable generalization of Algorithm \ref{alg:dp} can compute sequences reaching a goal sequence that maximizes the total payoff, but fails for any other goal sequence. However, since fairness concerns come into play with asymmetric payoffs, maximizing the total payoff may not be sufficient, requiring novel algorithmic techniques.  

Another possible extension is to consider {\em Bayesian} games \cite{bayesian}, where agents have private information about how they value certain outcomes.  In such games, play might differ depending on an agent's type, and we may even have defection and restarting on the path of play in equilibrium.  For example, if a small fraction of the population obtains great benefit from deviating in the first round, it may be better to tolerate those few agents repeatedly taking advantage for one round and then finding another partner, than to add significant hazing for the entire population to prevent this.  In this context, {\em partial} anonymity could also be of interest, for example where one can {\em choose} to reveal one's history to a new partner, for example to show that one behaved properly in a previous relationship but the partner unfairly defected.


\section*{Acknowledgments}
We thank the Cooperative AI Foundation, Polaris Ventures
(formerly the Center for Emerging Risk Research), and Jaan
Tallinn’s donor-advised fund at Founders Pledge for financial support. We are grateful to Matt Rognlie for his valuable contributions to the conceptualization of and preliminary work on this project. We thank Anupam Gupta for his helpful comments on FPTAS literature. Lastly, we thank the members of Foundations of Cooperative AI Lab (FOCAL) for comments and suggestions on earlier drafts of this work.

\bibliographystyle{named}
\bibliography{refs}

\begin{thebibliography}{}

\bibitem[\protect\citeauthoryear{Andersen and Conitzer}{2013}]{Andersen13:Fast}
Garrett Andersen and Vincent Conitzer.
\newblock Fast equilibrium computation for infinitely repeated games.
\newblock In {\em Proceedings of the Twenty-Seventh AAAI Conference on Artificial Intelligence}, pages 53--59, Bellevue, WA, USA, 2013.

\bibitem[\protect\citeauthoryear{Borgs \bgroup \em et al.\egroup }{2010}]{Borgs10:Myth}
Christian Borgs, Jennifer Chayes, Nicole Immorlica, Adam~Tauman Kalai, Vahab Mirrokni, and Christos Papadimitriou.
\newblock The myth of the {F}olk {T}heorem.
\newblock {\em Games and Economic Behavior}, 70(1):34--43, 2010.

\bibitem[\protect\citeauthoryear{Conitzer and Oesterheld}{2023}]{Conitzer23:Foundations}
Vincent Conitzer and Caspar Oesterheld.
\newblock Foundations of cooperative {AI}.
\newblock In {\em Proceedings of the Thirty-Seventh AAAI Conference on Artificial Intelligence}, pages 15359--15367, Washington, DC, USA, 2023.

\bibitem[\protect\citeauthoryear{Conitzer and Yokoo}{2010}]{Conitzer10:Using}
Vincent Conitzer and Makoto Yokoo.
\newblock Using mechanism design to prevent false-name manipulations.
\newblock {\em AI Magazine}, 31(4):65--77, 2010.

\bibitem[\protect\citeauthoryear{Dafoe \bgroup \em et al.\egroup }{2021}]{Dafoe21:Cooperative}
Allan Dafoe, Yoram Bachrach, Gillian Hadfield, Eric Horvitz, Kate Larson, and Thore Graepel.
\newblock Cooperative {AI}: machines must learn to find common ground.
\newblock {\em Nature}, 593(7857):33--36, 2021.

\bibitem[\protect\citeauthoryear{Datta}{1996}]{Saikat:lender_borrower}
Saikat Datta.
\newblock Building trust.
\newblock Sticerd - theoretical economics paper series, Suntory and Toyota International Centres for Economics and Related Disciplines, LSE, 1996.

\bibitem[\protect\citeauthoryear{Fujiwara-Greve and Okuno-Fujiwara}{2009}]{Fujiwara-Greve:VoluntaryPD}
Takako Fujiwara-Greve and Masahiro Okuno-Fujiwara.
\newblock Voluntarily separable repeated prisoner's dilemma.
\newblock {\em The Review of Economic Studies}, 76(3):993--1021, 2009.

\bibitem[\protect\citeauthoryear{Ghosh and Ray}{1996}]{Ghosh:Repeated}
Parikshit Ghosh and Debraj Ray.
\newblock Cooperation in community interaction without information flows.
\newblock {\em The Review of Economic Studies}, 63(3):491--519, 1996.

\bibitem[\protect\citeauthoryear{Harsanyi}{1967}]{bayesian}
John~C. Harsanyi.
\newblock Games with incomplete information played by {``Bayesian''} players, i-iii. part i. the basic model.
\newblock {\em Management Science}, 14(3):159--182, 1967.

\bibitem[\protect\citeauthoryear{Ibarra and Kim}{1975}]{ibarra_kim_FPTAS}
Oscar~H. Ibarra and Chul~E. Kim.
\newblock Fast approximation algorithms for the knapsack and sum of subset problems.
\newblock {\em J. ACM}, 22(4):463–468, oct 1975.

\bibitem[\protect\citeauthoryear{Izquierdo \bgroup \em et al.\egroup }{2014}]{Izquierdo:Leave}
Luis~R. Izquierdo, Segismundo~S. Izquierdo, and Fernando Vega-Redondo.
\newblock Leave and let leave: A sufficient condition to explain the evolutionary emergence of cooperation.
\newblock {\em Journal of Economic Dynamics and Control}, 46:91--113, 2014.

\bibitem[\protect\citeauthoryear{Jansen and Kraft}{2018}]{Jansen:FPTAS}
Klaus Jansen and Stefan~E.J. Kraft.
\newblock A faster {FPTAS} for the unbounded knapsack problem.
\newblock {\em European Journal of Combinatorics}, 68:148--174, 2018.
\newblock Combinatorial Algorithms, Dedicated to the Memory of Mirka Miller.

\bibitem[\protect\citeauthoryear{Kontogiannis and Spirakis}{2008}]{Kontogiannis08:Equilibrium}
Spyros~C. Kontogiannis and Paul~G. Spirakis.
\newblock Equilibrium points in fear of correlated threats.
\newblock In {\em Proceedings of the Fourth Workshop on Internet and Network Economics (WINE)}, pages 210--221, Shanghai, China, 2008.

\bibitem[\protect\citeauthoryear{Kranton}{1996}]{Kranton:repeated}
Rachel~E Kranton.
\newblock {The Formation of Cooperative Relationships}.
\newblock {\em The Journal of Law, Economics, and Organization}, 12(1):214--233, April 1996.

\bibitem[\protect\citeauthoryear{Lawler}{1979}]{Lawler:FPTAS}
Eugene~L. Lawler.
\newblock Fast approximation algorithms for knapsack problems.
\newblock {\em Mathematics of Operations Research}, 4(4):339--356, 1979.

\bibitem[\protect\citeauthoryear{Littman and Stone}{2005}]{Littman05:Polynomial}
Michael~L. Littman and Peter Stone.
\newblock A polynomial-time {N}ash equilibrium algorithm for repeated games.
\newblock {\em Decision Support Systems}, 39:55--66, 2005.

\bibitem[\protect\citeauthoryear{Lueker}{1975}]{Lueker:TwoNP}
G.S. Lueker.
\newblock Two {NP}-complete problems in nonnegative integer programming.
\newblock Technical Report 178, Computer Science Laboratory, Princeton University, 1975.

\bibitem[\protect\citeauthoryear{Magazine and Chern}{1984}]{Magazine:MultiFPTAS}
Michael~J. Magazine and Maw-Sheng Chern.
\newblock A note on approximation schemes for multidimensional knapsack problems.
\newblock {\em Mathematics of Operations Research}, 9(2):244--247, 1984.

\bibitem[\protect\citeauthoryear{Rob and Yang}{2010}]{RobYan:LongTerm}
Rafael Rob and Huanxing Yang.
\newblock Long-term relationships as safeguards.
\newblock {\em Economic Theory}, 43(2):143--166, 2010.

\bibitem[\protect\citeauthoryear{Rudin}{1976}]{rudin1976principles}
Walter Rudin.
\newblock {\em Principles of mathematical analysis}, volume~3.
\newblock McGraw-hill New York, 1976.

\bibitem[\protect\citeauthoryear{Wei}{2019}]{WEI2019311}
Dong Wei.
\newblock A model of trust building with anonymous re-matching.
\newblock {\em Journal of Economic Behavior \& Organization}, 158:311--327, 2019.

\bibitem[\protect\citeauthoryear{Wojtczak}{2018}]{rationalNPhard}
Dominik Wojtczak.
\newblock On strong {NP}-completeness of rational problems.
\newblock In Fedor~V. Fomin and Vladimir~V. Podolskii, editors, {\em Computer Science -- Theory and Applications}, pages 308--320, Cham, 2018. Springer International Publishing.

\bibitem[\protect\citeauthoryear{Yokoo \bgroup \em et al.\egroup }{2001}]{Yokoo01:Robust}
Makoto Yokoo, Yuko Sakurai, and Shigeo Matsubara.
\newblock Robust combinatorial auction protocol against false-name bids.
\newblock {\em Artificial Intelligence}, 130(2):167--181, 2001.

\bibitem[\protect\citeauthoryear{Yokoo \bgroup \em et al.\egroup }{2004}]{Yokoo04:Effect}
Makoto Yokoo, Yuko Sakurai, and Shigeo Matsubara.
\newblock The effect of false-name bids in combinatorial auctions: New fraud in {Internet} auctions.
\newblock {\em Games and Economic Behavior}, 46(1):174--188, 2004.

\end{thebibliography}
\appendix
\clearpage
\section{Missing Proofs}
\subsection{Proof of Claim \ref{claim:closed} and Example for Non-Uniqueness}\label{appendix:lemma1_proofs}
Below we prove Claim \ref{claim:closed} from Proposition \ref{prop:opt_existence}. For an example demonstrating that the optimal sequence is not necessarily unique, see Example \ref{eg:nonunique} below.

        Given a game $G=(\{p^{(j)},p^{*(j)}\}_{j\in [n]}, \beta)$, say $T$ is the set of achievable total discounted utilities. Fix some $(t_k)_{k \in \mathbb{N}} \in T^\mathbb{N}$ that converges to a limit $t$. To prove closedness, we must show that $t$ is achievable. For each $k\in \mathbb{N}$, there exists a stable sequence $(p^k_i)_{i\in \mathbb{N}}$ with total discounted utility $t_k$. Given this sequence of sequences  $\left((p^k_i)_{i\in \mathbb{N}}\right)_{k \in \mathbb{N}}$, we construct a new sequence $(p'_i)_{i \in \mathbb{N}}$ in the following iterative manner:
        \begin{itemize}
            \item \textbf{Step 0:} Since there are finitely many choices for $p^k_0$ (the 0th element of each sequence), there needs to be at least one $p \in \{p^{(j)}\}_{j \in [n]}$ such that $p^k_0=p$ for infinitely many $k$. Choose the smallest $j^* \in [n]$ such that $p^{(j^*)}$ fulfills this condition, and set $p'_0= p^{(j^*)}$. Let $\left((p^{n_0(k)}_i)_{i\in \mathbb{N}}\right)_{k \in \mathbb{N}} \subset \left((p^k_i)_{i\in \mathbb{N}}\right)_{k \in \mathbb{N}}$ be the subsequence of sequences such that $p^{n_0(k)}_0= p^{(j^*)}$ for all $k$. 
            \item  \textbf{Step }$\boldsymbol{\ell>0:}$ Note that after step $\ell-1$, we are given a sequence of sequences $\left((p^{n_{\ell-1}(k)}_i)_{i\in \mathbb{N}}\right)_{k \in \mathbb{N}} $ where each $(p^{n_{\ell-1}(k)}_i)_{i\in \mathbb{N}}$ agree on steps $i=0,...,\ell-1$. We now repeat the same process as Step 0, picking the smallest $j^*$ such that $p^{n_{\ell-1}(k)}_\ell= p^{(j^*)}$ for infinitely many $k$. We then set $p'_\ell= p^{(j^*)}$ and let $\left((p^{n_{\ell}(k)}_i)_{i\in \mathbb{N}}\right)_{k \in \mathbb{N}} \subset \left((p^{n_{\ell-1}(k)}_i)_{i\in \mathbb{N}}\right)_{k \in \mathbb{N}}  $ be the subsequence of sequences such that $p^{n_{\ell}(k)}_\ell= p^{(j^*)}$ for all $k_\ell$. 
        \end{itemize}
        We argue (a) $(p'_i)_{i\in \mathbb{N}}$ is stable and (b) $\sum_{i=0}^\infty \beta^i p'_i=t$. For (b), assume for the sake of contradiction that $\sum_{i=0}^\infty \beta^i p'_i=t + \delta$ for some $\delta \neq 0$. Let $p_{amax} = \max_{j \in [n]}\left|p^{(j)}\right|$ and choose $\ell$ large enough such that $2 \beta^{\ell+1}\frac{p_{amax}}{1-\beta}< \delta$ and let $\varepsilon = \delta - 2 \beta^{\ell+1}\frac{p_{amax}}{1-\beta}$. Let $\left((p^{n_\ell(k)}_i)_{i\in \mathbb{N}}\right)_{k \in \mathbb{N}}$ be the sequence of sequences achieved after step $\ell$ above, and $(t_{n_\ell(k)})_{k \in \mathbb{N}}$ the corresponding total utilities. Since $\left((p^{n_\ell(k)}_i)_{i\in \mathbb{N}}\right)_{k \in \mathbb{N}}$ is a subsequence of $\left((p^{k}_i)_{i\in \mathbb{N}}\right)_{k \in \mathbb{N}}$, then $(t_{n_\ell(k)})_{k \in \mathbb{N}}$ is a subsequence of $(t_k)_{k \in \mathbb{N}}$, hence it also converges to $t$. Thus, we can choose $k'$ large enough such that $|t_{n_\ell(k')}-t| < \varepsilon$. Since $p'_i= p^{n_\ell(k')}_i$ for $i = 0,...,\ell$, we have:
        \begin{align*}
            \delta &= \sum_{i=0}^\infty \beta^i p'_i - t =\left( \sum_{i=0}^\infty \beta^i p'_i -t_{n_\ell(k')}\right)+( t_{n_\ell(k')}- t)\\
            & = \left( \sum_{i=\ell+1}^\infty \beta^i (p'_i -p^{n_\ell(k')}_i)\right)+( t_{n_\ell(k')}- t)\\
            & < \left( \sum_{i=\ell+1}^\infty \beta^i 2p_{amax}\right)+\varepsilon 
            = 2\beta^{\ell+1} \frac{p_{amax}}{1-\beta} +\varepsilon=\delta
        \end{align*}
        which is a contradiction since it implies $\delta < \delta$. Hence, we have indeed have $\sum_{i=0}^\infty \beta^i p'_i=t$. We now show (a), i.e., that $(p'_i)_{i \in \mathbb{N}}$ is indeed stable. For the sake of contradiction, assume it is not stable, which implies that there exists same $\ell  \in \mathbb{N}$ such that:
        \begin{align*}
        \sum_{i=0}^{\ell-1} \beta^i p'_i + \beta^\ell p'^*_\ell > (1-\beta^{\ell+1}) \sum_{i=0}^\infty \beta^i p'_i = (1-\beta^{\ell+1}) t
        \end{align*}
        Define $\delta' =  \sum_{i=0}^{\ell-1} \beta^i p'_i + \beta^\ell p'^*_\ell - (1-\beta^{\ell+1}) t$ and let $\left((p^{n_\ell(k)}_i)_{i\in \mathbb{N}}\right)_{k \in \mathbb{N}}$ be the sequence of sequences achieved after step $\ell$ above, and $(t_{n_\ell(k)})_{k \in \mathbb{N}}$ the corresponding total utilities. Since $(t_{n_\ell(k)})_{k \in \mathbb{N}}$ converges to $t$, we can choose $k'$ large enough such that $|t- t_{n_\ell(k')}| < \frac{\delta'}{(1-\beta^{\ell+1})}$,
        but then:
        \begin{align*}
            \sum_{i=0}^{\ell-1 } &\beta^ip^{n_\ell(k')}_i +\beta^\ell p^{* n_\ell(k')}_\ell =   \sum_{i=0}^{\ell-1} \beta^i p'_i + \beta^\ell p'^*_\ell \\&= \delta' + (1-\beta^{\ell+1}) t \\
            &=  \delta'+ (1-\beta^{\ell+1}) (t- t_{n_\ell(k')})+ (1-\beta^{\ell+1}) t_{n_\ell(k')}\\
            &\geq  \delta'- (1-\beta^{\ell+1}) |t- t_{n_\ell(k')}|+ (1-\beta^{\ell+1}) t_{n_\ell(k')}\\
            & >  \delta'- \frac{(1-\beta^{\ell+1}) \delta'}{(1-\beta^{\ell+1})}+ (1-\beta^{\ell+1}) t_{n_\ell(k')}\\
            &=  (1-\beta^{\ell+1}) t_{n_\ell(k')}
        \end{align*}
        which contradicts the stability of $(p^{n_\ell(k')}_i)_{i\in \mathbb{N}}$ at step $\ell$. Since we started from a sequence of stable sequences, this is a contradiction, implying that $(p'_i)_{i \in \mathbb{N}}$ is indeed stable. Hence, $t$ is indeed an achievable total discounted utility, implying $t \in T$ and proving the closedness of $T$. 

        We now present an example showing that the optimal sequence for a given game and $\beta$ is not always unique. 
        \hfill \qedsymbol
\begin{example}[Non-uniqueness] 
\label{eg:nonunique}
Consider the game below, where the payoffs are listed for the row and the columns players, respectively. 
    \begin{align*}
            \centering
            \begin{tabular}{|c|c|c|c|c|c|c|c|}
            \hline & $D$& $C_1$ &$C_2$& $C_3$ &$C_4$&$C_5$\\
            \hline
            $D$ & 0,0 & 6,0 & 6,0  & 8,0 & 9,0&10,0\\
            \hline $C_1$ & 0,6& 4.5,4.5 & 0,0 & 0,0 & 0,0  & 0,0 \\
            \hline $C_2$ & 0,6&  0,0 & 5,5 & 0,0 & 0,0  & 0,0 \\
            \hline $C_3$ & 0,8&  0,0 & 0,0 & 6,6 & 0,0  & 0,0 \\
            \hline $C_4$ & 0,9&  0,0 & 0,0 & 0,0 &  7,7 & 0,0 \\
            \hline $C_5$ & 0,10&  0,0 & 0,0 & 0,0 &   0,0 & 8,8 \\\hline
            \end{tabular}
    \end{align*}
    If $\beta=1/2$, both of the following sequences are stable and optimal:
     \begin{align*}
         (a^1_i)_{i \in \mathbb{N}} &= (C_2,C_3,C_5,C_5,C_5,C_5,\ldots)\\
         (p^1_i)_{i \in \mathbb{N}} &= (5,6,8,8,8,8,\ldots)\\
         (p^{1*}_i)_{i \in \mathbb{N}} &= (6,8,10,10,10,10,\ldots),
     \end{align*}
     and 
     \begin{align*}
        (a^2_i)_{i \in \mathbb{N}} &= (C_1,C_4,C_5,C_5,C_5,C_5,\ldots)\\
         (p^2_i)_{i \in \mathbb{N}} &= (4.5,7,8,8,8,8,\ldots)\\
         (p^{2*}_i)_{i \in \mathbb{N}} &= (6,9,10,10,10,10,\ldots),
     \end{align*}
both with total discounted utility 12. Stability can be shown via Equation (\ref{eq:stable}). Assume for the sake of contradiction that there exists $(p'_i)_{i \in \mathbb{N}}$ with total discounted utility larger than 12. This implies $p'_0+\beta p'_1>8$, since the total discounted utility from step $i=2$ onwards is at most $4$. Stability at $i=0$ implies $p'^*_0 \leq (1-\beta) \sum_{i=0}^\infty {\beta^i p_i} \leq (0.5) (p_0 + 8)$. Of the 6 available actions, this equation is only fulfilled by $D, C_1,$ and $C_2$. However, $p'_0+\beta p'_1> 8$ implies $p'_0> 4$, so $p'_0 \in \{4.5,5\}$. If $p'_0=4.5$, we must have $p'_1=8$, but this is not stable since it would imply $p'_0+ \beta p'^*_1 = 9.5 > (1-0.25) (4.5+8) \geq (1-\beta^2)\left(\sum_{i=0}^\infty \beta^i p_i\right)$. If $p'_0=5$, then $p'_1=8$ would similarly make the sequence unstable, so we must have $p'_1=7$ since $p'_0+\beta p'_1>8$. But then  $p'_0+ \beta p'^*_1 = 9.5 > (1-0.25) (5+3.5+4) \geq (1-\beta^2)\left(\sum_{i=0}^\infty \beta^i p_i\right), $ so once again stability is violated. Hence, 12 is the optimal total discounted utility. 

\end{example}
\subsection{Proofs of Lemmas \ref{lemma:repeating} and \ref{lemma:optimalgoal}} \label{appendix:goal_lemmas}
\subsubsection{Proof of Lemma \ref{lemma:repeating}}
     Take any stable $(p_i)_{i \in \mathbb{N}}$ that never reaches an infinitely repeating payoff, or reaches one that is not the highest in the sequence. We will show that it is not optimal. Let $p_\Omega$ be the largest payoff that appears in $(p_i)_{i \in \mathbb{N}}$ (well-defined since $\{p^{(j)}\}_{j\in [n]}$ is finite) and let $i^*$ be the index it first appears. We define a new sequence $(p'_i)_{i \in \mathbb{N}}$ as: \begin{align*}p'_i = \begin{cases} p_i & \text{for }i \leq i^* \\ p_\Omega & \text{otherwise} \end{cases}\end{align*} By assumption, there exists some $i> i^*$ with $p_i < p_\Omega = p'_i$, implying $\sum_{i=0}^\infty \beta^i p'_i > \sum_{i=0}^\infty \beta^i p_i$. We prove  the stability of $(p'_i)_{i \in \mathbb{N}}$ by showing that it fulfills Equation (\ref{eq:stable}) for each $k \in \mathbb{N}$.\\
     \noindent\textbf{For }$\boldsymbol{k \leq i^*}$: 
        \begin{align*}
            \sum_{i=0}^{k-1} \beta^i p'_i + \beta^k p'^*_k &=   \sum_{i=0}^{k-1} \beta^i p_i + \beta^k p^*_k \\&
            \leq (1-\beta^{k+1}) \sum_{i=0}^\infty \beta^i p_i  \\&< (1-\beta^{k+1}) \sum_{i=0}^\infty \beta^i p'_i 
        \end{align*}
        where the first inequality follows from the stability of $(p_i)_{i \in \mathbb{N}}$ at step $k$.  \\
        \noindent\textbf{For }$\boldsymbol{k > i^*}$:  Stability at $i^*$ (shown above) implies:
        \begin{align}
            \sum_{i=0}^{i^*-1} \beta^i p'_i+ \beta^{i^*} p^*_\Omega
            \leq (1-\beta^{i^*+1}) \sum_{i=0}^\infty \beta^{i} p'_i 
        \end{align}
        where we use the fact that $p'^*_i = p^*_\Omega$ for all $i \geq i^*$. Subtracting $\sum_{i=0}^{i^*-1} \beta^i p'_i$ and multiplying both sides by $\beta^{k-i^*}$ gives us:
        \begin{align}
            \beta^{k} p^*_\Omega \leq \sum_{i=i^*}^\infty \beta^{i-i^*+k} p'_{i} - \beta^{k+1} \left(\sum_{i=0}^\infty\beta^i p'_i \right)  \label{eq:foreverstable}
        \end{align}
        Notice that $\sum_{i=i^*}^\infty \beta^{i-i^*+k} p'_{i}  =\sum_{i=k}^\infty \beta^{i} p'_{i+i^*-k} = \sum_{i=k}^\infty \beta^{i} p'_{i}$ since for $i\geq k>i^*$, we have $p'_i= p_\Omega = p'_{i+i^{*}-k}$. Noting that $p^*_\Omega= p'^*_k$ and adding $\sum_{i=0}^{k-1} \beta^i p'_i$ to both sides of the equation gives us the stability result at step $k$. Intuitively this reflects that for every step after $i^*$, the decision to deviate or not is equivalent.\\
    Hence, $(p'_i)_{i \in \mathbb{N}}$ is stable and surpasses $(p_i)_{i \in \mathbb{N}}$, proving the non-optimality of the latter.
    \hfill\qedsymbol
\subsubsection{Proof of Lemma \ref{lemma:optimalgoal}}
   We prove the lemma in two steps: First we show that if there is any stable sequence for some $\beta$ (if not, the lemma is vacuously true), then for sufficiently large $\beta$ there is a stable sequence with $p_\Omega$ as the goal value. Second, we show there exists a $\beta'$ such that for all $\beta>\beta'$, the total discounted utility of the sequence is strictly larger than that of any other sequence with a smaller goal value.  
   
   Assume $(p_i)_{i \in \mathbb{N}}$ is a stable sequence for some $\beta \in [0,1)$. By the proof of Lemma \ref{lemma:repeating}, we can assume that $(p_i)_{i \in \mathbb{N}}$ converges to an infinitely repeating goal value $p$ (with $p\geq p_i$ for all $i$). Note that if $p= p_\Omega$, then we are done. Otherwise, by stability of $(p_i)_{i \in \mathbb{N}}$ at step $k=0$, we have:
    \begin{align}
    p^*_0  \leq (1-\beta) \sum_{i=0}^\infty \beta^i p_i  \leq (1-\beta) \sum_{i=0}^\infty \beta^i p = p < p_\Omega
    \end{align}
     Since $p_0^* < p_\Omega$, we can choose $a =\max\left( \left \lfloor \frac{p^*_\Omega-p_\Omega}{ p_\Omega-p_0^*}\right\rfloor+1, 1\right)$. We define a new sequence $(p'_i)_{i\in \mathbb{N}}$ as $p'_i = \begin{cases} p_0 &\text{for }i<a\\ p_\Omega &\text{ otherwise}\end{cases}$. Say $P' = \sum_{i=0}^\infty \beta^i p'_i$. We now prove the stability of this sequence, showing it fulfills (\ref{eq:stable}) for each $k$ for sufficiently large $\beta$:\\
        \noindent \textbf{For }$\boldsymbol{k < a}$: Note that as $\beta\rightarrow 1$, we have $(1-\beta)P'= (1-\beta^a) p_0 + \beta^a p_\Omega \rightarrow p_\Omega > p^*_0=p'^*_k$. Hence we can choose $\beta$ large enough that $p'^*_k \leq (1-\beta)P' \leq \sum_{i=0}^\infty p'_{i+k} \beta^i-\beta P'$ where the second inequality follows from the fact that $(p'_i)_{i \in \mathbb{N}}$ is non-decreasing. Multiplying both sides by $\beta^k$ and adding $\sum_{i=0}^{k-1}\beta^i p'_i$ gives the stability condition at step $k$.\\
         \noindent \textbf{For }$\boldsymbol{k \geq a}$: 
         Note that as $\beta \rightarrow 1$, we have $\beta \left(\frac{p_\Omega}{1-\beta} - P' \right)=\beta \sum_{i=0}^{a-1} \beta^i (p_\Omega - p'_i)  \rightarrow a \cdot (p_\Omega - p_0) > p^*_\Omega - p_\Omega$. Hence for sufficiently large $\beta$, we will have $p^*_\Omega + \beta P' < p_\Omega + \beta \frac{p_\Omega}{1-\beta} =   \frac{p_\Omega}{1-\beta}$. Since $k\geq a$ we have $p'^*_k = p^*_\Omega$ and $p_i = p_\Omega$ for $i\geq k$. Hence, multiplying both sides by $\beta^k$ and adding $\sum_{i=0}^{k-1}\beta^i p_i$ gives us the stability condition at step $k$.

       Next, we show that for large enough $\beta$, $(p'_i)_{i\in \mathbb{N}}$ will surpass any sequence without $p_\Omega$ as different goal value. Say $p_{\Psi}=\max \{p^{(j)}| j \in [n] , p^{(j)} < p_\Omega\}$ is the largest payoff after $p_\Omega$. As shown above, as $\beta \rightarrow 1$, we have $(1-\beta)P' \rightarrow p_\Omega > p_\Psi$, so there exists $\beta'$ such that for all $\beta>\beta'$, we have $P' > \frac{p_\Psi}{1-\beta}$. Any sequence that does not contain $p_\Omega$ will have its total utility bound by $\frac{p_\Psi}{1-\beta}$ and hence is surpassed by $(p'_i)_{i\in \mathbb{N}}$ for $\beta > \beta'$. Any sequence that contains $p_\Omega$ but does not converge to it is non-optimal by Lemma \ref{lemma:repeating}. This proves the non-optimality of any sequence that does not have goal value $p_\Omega$ for $\beta>\beta'$, completing the proof of the lemma.
       \hfill\qedsymbol
\subsection{Proof of Proposition \ref{prop:hstability}}\label{appendix:hstability}
We investigate the stability condition given in Equation (\ref{eq:stable}) in the $\beta \rightarrow 1$ limit and for sequences with goal value $p_\Omega$. Given a sequence $(p_i)_{i \in \mathbb{N}}$ with $p_i = p_\Omega$ for all $i \geq z$, the stability condition at step $k$ is:
\begin{align}
        \sum_{i=0}^{k-1} \beta^i p_i+ \beta^k p^*_k 
 \leq (1-\beta^{k+1}) \left(\sum_{i=0}^{z-1} \beta^{i} p_i  + \frac{\beta^z p_\Omega}{1-\beta} \right) \label{eq:repeatingstability}
\end{align}
Taking the $\beta \rightarrow 1$ limit, the left hand side of (\ref{eq:repeatingstability}) becomes $\sum_{i=0}^{k-1}p_i+p^*_k$, and the right hand side  (using L'Hôpital's rule)  becomes $(k+1) p_\Omega$. Hence, stability is satisfied if and only if $\sum_{i=0}^{k-1}p_i+p^*_k < (k+1) p_\Omega$, or equivalently:
\begin{align}
    \sum_{i=0}^{k-1} (p_\Omega - p_i) > p^*_k - p_\Omega 
\end{align}
The inequality is strict since for sufficiently large $\beta$, the derivative of the left minus right hand sides of (\ref{eq:repeatingstability}) with respect to $\beta$ is negative, so if the limits of the two sides are equal, there is a $\beta'$ such that the stability fails for all $\beta > \beta'$. Intuitively, the inequality needs to be strict because for any $\beta<1$, a ``tie'' between the two sides of the equation would be broken towards the side of deviation, as that deviation payoff comes earlier.
\hfill\qedsymbol

\subsection{Proofs of Claims from Theorem \ref{thm:fptas}}\label{appendix:thm3_proofs}
     \subsubsection{Proof of Claim \ref{claim:tilde}}
    Before Line \ref{line:tilde}, Algorithm \ref{alg:fptas} can only halt on Line \ref{line:error} or \ref{line:no_halt}. If it halts on Line \ref{line:error}, this implies $t^{(j)} \geq 0$ for all $j\in [n]$, so no action will satisfy (\ref{eq:hstability}) for $k=0$, implying there is no stable sequence at all, and the algorithm appropriately raises an error. Otherwise, the set $F$, assigned on Line \ref{line:eligible}, contains all legal first actions that do not meet the final hazing threshold alone. If the algorithm halts on Line \ref{line:no_halt}, then $F$ is empty, meaning that any legal first action will also be the final action in the hazing sequence, since it will have hazing cost at least $\Delta+1$. Then the optimal sequence is simply the legal action with the least hazing cost, which is what is assigned to $\hat{H}$ on Line \ref{line:pre_no_halt}, implying $\hat{H}=H^*$ as claimed, where $H^*$ is the optimal hazing. If the algorithm has not halted by Line  \ref{line:tilde}, this implies $F \neq \emptyset$. We then have $\Tilde{H} = k \cdot h^{(j^*)}$ for some $j^* \in F$. Notice that the hazing sequence $(h_{i})_{i\in \{0\}\cup [k-1]}$ where $h_i =h^{(j^*)}$ is stable at every step (since $t^{(j^*)}<0$) and has total hazing cost $k \cdot h^{(j^*)} > \Delta$, implying it is a stable sequence. Then by the optimality of $H^*$, we have $H^* \leq \Tilde{H}$. Since $k$ is the minimum integer s.t. $k \cdot h^{(j^*)} > \Delta$ (and $k \geq 2$ by definition of $F$), we have $H^* \geq \Delta +1 > (k-1) h^{(j^*)} \geq (k-k/2) h^{(j^*)}= \frac{\Tilde{H}}{2}$, completing the proof of the claim. 
    \hfill \qedsymbol

\subsubsection{Proof of Claim \ref{claim:bounds}}
        We have $ f^{(j)} = \left\lfloor \frac{h^{(j)} }{\delta} \right\rfloor \leq   \frac{h^{(j)} }{\delta} \Rightarrow  f^{(j)} \cdot \delta \leq h^{(j)}$. Moreover, notice that any action with $h^{(j)}\leq \frac{\varepsilon}{3} \Tilde{H}$ is added to $S$ in Line \ref{line:small}, so any action that does not get added to $S$ (i.e., any action acted upon by Lines \ref{line:normal_hazing}-\ref{line:fptas_endfor}) has $\frac{h^{(j)} }{\delta}>  \frac{\varepsilon}{3} \Tilde{H} / \delta = \frac{3}{\varepsilon}$. Therefore, $f^{(j)} = \left\lfloor \frac{h^{(j)} }{\delta} \right\rfloor >  \frac{h^{(j)} }{\delta}-1 > \frac{3-\varepsilon}{\varepsilon}$, implying:
        \begin{align}
            h^{(j)} &\leq \delta (f^{(j)}+1) = f^{(j)}\cdot  \delta (1+1/f^{(j)})\\& < f^{(j)}\cdot  \delta \left(1+\frac{\varepsilon}{3-\varepsilon}\right) \leq  f^{(j)}\cdot \delta \left(1+\frac{\varepsilon}{2}\right)
        \end{align}
        since $\varepsilon \leq 1$. 
\hfill\qedsymbol

\subsubsection{Proof of Claim \ref{claim:agreement}}
It is sufficient to induct on each execution of the lines where $T$ or $L$ are edited, which are Lines $\ref{line:prefor}$ and $\ref{line:pre_fptas_endfor}-\ref{line:fptas_endfor}$. Notice that in Line $\ref{line:prefor}$ (which gets executed once), all $T[k]$ and $L[k]$ are $\textit{inf}$ except for $L[0]=[]$ and $T[0]=0$, which satisfies the sums in the claim. Using this as the base case, assume that $T$ and $L$ satisfies the claim before an execution of Lines $\ref{line:pre_fptas_endfor}-\ref{line:fptas_endfor}$; we want to show that the claim is still satisfied after this single execution. We know that $T[k] \neq \textit{inf}$ by Line \ref{line:substability}. Hence, using the inductive hypothesis, after Line $\ref{line:pre_fptas_endfor}-\ref{line:fptas_endfor}$ are executed, we will have $\sum_{j' \in L[k+f^{(j)}]} f^{(j')}= \sum_{j' \in L[k]+ [j]} f^{(j')}= \sum_{j'\in L[k]} f^{(j')} + f^{(j)} = k + f^{(j)}$ and $\sum_{j' \in L[k+f^{(j)}]} h^{(j')}= \sum_{j' \in L[k]+ [j]} h^{(j)'}= \sum_{j'\in L[k]} h^{(j')} + h^{(j)} = T[k] + h^{(j)}=T[k+f^{(j)}]$, as desired. Since all other entries of $T$ and $L$ remain unchanged after this single execution, the lists continue to satisfy the claim. 
\hfill\qedsymbol

\subsubsection{Proof of Claim \ref{claim:overshoot}}
We first present and prove another claim: 
   \begin{claim} \label{claim:mono}
    Say that at some step $T[k]$ is set to some non-\textit{inf} value $a$. Then after the execution of the algorithm is completed, we will have $T[k] \geq a$ for some finite $T[k]$. 
\end{claim}
\begin{proof}[Proof of Claim \ref{claim:mono}]
        Notice that once any $T[k]$ is set to a non-\textit{inf} value only modification to $T$ during the algorithm (Line \ref{line:pre_fptas_endfor}) is conditional on the new value being strictly larger (the if statement in Line \ref{line:substability}); hence the claim follows. 
\end{proof}

Now, we go back to the proof of Claim \ref{claim:overshoot}. By assumption, the actions appear in order $1,...,n$ in $(h_i)_{i\in[\ell]} \in \{h^{(j)}\}_{j \in \{0\} \cup [n]}^\ell$. Say action $(j)$ appears $r_j\geq 0$ times. Say $T^z$ is the list after $z$ executions of the for loop in Line \ref{line:outer}, with $T^0$ defined as in Line \ref{line:prefor}.

     We claim that for all $z\in \{0,...,n\}$, we have $T^z\left[\sum_{j=1}^{z}f^{(j)}r_j\right] \geq \sum_{j=1}^{z}h^{(j)}r_j$, where the LHS is not \textit{inf}. We will prove this statement by strong induction. As a base case, we have $z=0$, we have $T^0[0]=0$, so the inductive claim is vacuously correct. Now, assuming that the inductive claim is true for $z=j-1$ for some $j\in \{1,...,n\}$, we would like to prove it for $z=j$. Notice that if $r_j =0$, by the inductive hypothesis and by Claim \ref{claim:mono}, we have $T^j\left[\sum_{j'=1}^{j}f^{(j')}r_{j'}\right]=T^j\left[\sum_{j'=1}^{j-1}f^{(j')}r_{j'}\right]\geq T^{j-1}\left[\sum_{j'=1}^{j-1}f^{(j')}r_{j'}\right] \geq \sum_{j'=1}^{j-1}h^{({j'})}r_{j'} = \sum_{j'=1}^{j}h^{({j'})}r_{j'}$, and the inductive claim is true. Otherwise, we have $r_j \geq 1$ and the if statement Line \ref{line:small} is false by assumption on $(h_i)_{i \in [\ell]}$. Consider the inner for loop in Line \ref{line:inner}. Notice that since $\sum_{i=0}^\ell h_i \leq \Tilde{H}$, we have $\sum_{j'=1}^{j}r_{j'} f^{(j')} \leq g$, hence $\sum_{j'=1}^{j-1}r_{j'} f^{(j')} + a \cdot f^{(j)} \leq g-f^{(j)}$ for all $a \in \{0,...,r_j-1\}$. Say $T^j_a$ is the final list after the $k=\sum_{j'=1}^{j-1}r_{j'} f^{(j')} +  a \cdot f^{(j)}$ iteration of the for loop in Line \ref{line:inner} is complete, with $T^j_{-1}$ being the list right before $k=\sum_{j'=1}^{j-1}r_{j'} f^{(j')} $ iteration starts. Notice we have $\textit{inf} \neq T^j_{-1}[k] \geq T^{j-1}[z] \geq \sum_{j'=1}^{j-1}h^{(j')}r_{j'} > t^{(j)}$ where three inequalities follow from Claim \ref{claim:mono}, the inductive hypothesis, and the stability of $(h_i)_{i \in \{0\}\cup  [\ell]}$,  respectively. If Lines \ref{line:pre_fptas_endfor}-\ref{line:fptas_endfor} get executed, then we assign $T_0^j[k+f^{(j)}] = T^j_{-1}[k]+h^{(j)} \geq \sum_{j'=1}^{j-1}h^{(j')}r_{j'} + h^{(j)}$. If they do not get executed, this implies we already have  $\textit{inf} \neq T^j_{-1}[k+f^{(j)}] \geq  T^j_0[k]+h^{(j)} \geq \sum_{j'=1}^{j-1}h^{(j')}r_{j'} + h^{(j)}$ and $T^j_0=T^j_{-1}$ since the list will not be modified, showing that $T^j_{0}[k+f^{(j)}] \geq \sum_{j'=1}^{j-1}h^{(j')}r_{j'} + h^{(j)}$ to be true regardless. Repeating this process for the $k=\sum_{j'=1}^{j-1}r_{j'} f^{(j')} +  a \cdot f^{(j)}$ iteration of the for loop for each $a \in \{1,...,r_j-1\}$ (each of which comes after $a=0$) ensures that $T^j_a[k+f^{(j)}] \geq  \sum_{j'=1}^{j-1}h^{(j')}r_{j'} + (a+1)\cdot h^{(j)}$ for each $a$. Then, using Claim \ref{claim:mono} and the previous inequality for $a= r_{j}-1$, we have:
    \begin{align}
        T^j\left[\sum_{j'=1}^{j}f^{(j')}r_{j'} \right] &\geq  T^j_{r_j-1}\left[\sum_{j'=1}^{j}f^{(j')}r_{j'} \right] \geq \sum_{j'=1}^{j}h^{(j')}r_{j'}
    \end{align}
    completing the inductive proof. 
    \hfill\qedsymbol

\end{document}